\documentclass[12pt]{article}
\usepackage{graphicx,array,amssymb,amsmath,psfrag,fullpage,multirow,caption}
\usepackage{psfrag,graphics,epsfig,multirow,multicol,times,color,soul}
\usepackage{setspace}
\usepackage{enumitem}
\usepackage{algorithm}
\usepackage{algorithmic}
\usepackage{hyperref}

\usepackage[agsmcite]{harvard}

\newcommand{\mC}{\ensuremath{\mathcal{C}}}

\citationmode{abbr}
\definecolor{webblue}{rgb}{0,0,1}
\definecolor{webgreen}{rgb}{0,.5,0}
\definecolor{webbrown}{rgb}{.6,0,0}
\title{Generalized Isotonic Regression}

\author{Ronny Luss\thanks{Department of Statistics and Operations Research, Tel Aviv University,
Tel Aviv, Israel. \texttt{ronnyluss@gmail.com}} \and Saharon
Rosset\thanks{Department of Statistics and Operations Research, Tel Aviv University,
Tel Aviv, Israel. \texttt{saharon@post.tau.ac.il}}}

\newcommand{\BEAS}{\begin{eqnarray*}}
\newcommand{\EEAS}{\end{eqnarray*}}
\newcommand{\BEA}{\begin{eqnarray}}
\newcommand{\EEA}{\end{eqnarray}}
\newcommand{\BEQ}{\begin{equation}}
\newcommand{\EEQ}{\end{equation}}
\newcommand{\BIT}{\begin{itemize}}
\newcommand{\EIT}{\end{itemize}}
\newcommand{\BNUM}{\begin{enumerate}}
\newcommand{\ENUM}{\end{enumerate}}

\newcommand{\BA}{\begin{array}}
\newcommand{\EA}{\end{array}}
\newcommand{\BC}{\begin{center}}
\newcommand{\EC}{\end{center}}

\newcommand{\reals}{{\mbox{\bf R}}}

\newcommand{\QED}{~~\rule[-1pt]{6pt}{6pt}}

\newcommand{\argmin}{\mathop{\rm argmin}}

\newtheorem{theorem}{Theorem}

\newtheorem{proposition}[theorem]{Proposition}

\newtheorem{remark}[theorem]{Remark}

\newcounter{exno}

\newenvironment{proof}{\textbf{Proof.}}{\QED\bigskip}

{\begin{quote}}{\end{quote}}

\newcounter{oursection}

\newcounter{lecture}

\begin{document}
\maketitle
\begin{abstract}
We present a new computational and statistical approach for fitting
isotonic models under convex differentiable loss functions through
recursive partitioning. Models along the partitioning path are also
isotonic and can be viewed as regularized solutions to the problem.
Our approach generalizes and subsumes the well-known work of Barlow
and Brunk (1972) on fitting isotonic regressions subject to
specially structured loss functions, and expands the range of loss
functions that can be used (for example, adding Huber's loss for
robust regression).  This is
accomplished through an algorithmic adjustment to the recursive
partitioning approach recently developed for solving large scale
$\l_2$-loss isotonic regression problems (Spouge et al. 2003, Luss
et al. 2011).  We prove that the new algorithm solves the generalized problem while maintaining the favorable computational and statistical properties of the $l_2$ algorithm.  The results are demonstrated on both real and synthetic data in two
settings: fitting count data using negative Poisson log-likelihood
loss, and fitting robust isotonic regressions using Huber's loss.
\end{abstract}

\doublespacing
\section{Introduction}
In this paper, we generalize recently developed algorithms for
solving large-scale isotonic regressions with $l_2$ loss function
\cite{Spou2003,Luss2012} in order to handle a  more general class of
loss functions.  These generalizations allow for fitting isotonic
regressions with useful loss functions such as Huber's loss, which was previously impractical for large problems using generic convex optimization solvers. For example, isotonic regression with Huber's loss can be solved with generic quadratic programming solvers that suffer due to the large number of constraints in our problems, whereas the algorithm we introduce takes advantage of the structured constraints and is more efficient by orders of magnitude. Isotonic regression is a nonparametric approach for building models whose fits are monotone
in their covariates. Such assumptions are natural to applications in
biology \cite{Oboz2008}, ranking \cite{Zhen2008}, medicine
\cite{Sche1997}, statistics \cite{Barl1972} and psychology
\cite{Kruskal64}.  Assume $n$ data observations $(x_1, y_1), ...,
(x_n,y_n)$ and a partial order $\preceq$, e.g., the standard
Euclidean one where $x_1 \preceq x_2$ if and only if $x_{1j} \leq
x_{2j}$ coordinate-wise. We index the set of isotonicity constraints
implied by the partial order by $\mathcal{I}=\{(i,j):x_i \preceq
x_j\}$.  Classic isotonic regression considers the $l_2$ loss
function and solves
\begin{equation}
\min{\{\displaystyle\sum_{i=1}^n{(\hat{y}_i-y_i)^2} : \hat{y}_i\leq
\hat{y}_j\quad\forall(i,j)\in\mathcal{I}\}}
\label{eq:IR}\end{equation} in $\hat{y}\in\reals^n$.  Throughout the
paper we denote by $m=|\mathcal{I}|$ the number of isotonic
constraints and $d$ the dimension of data, i.e., $x_i\in\reals^d$.

While the assumption of isotonicity is often natural,  isotonic
regression has not been extensively applied in ``modern''
applications for two main reasons.  As the number of observations
$n$, the data dimensionality $d$, and the number of isotonicity
constraints $m$ increase, problem (\ref{eq:IR}) suffers from
computational as well statistical (i.e., overfitting) difficulties.
These are reviewed in \citeasnoun{Luss2012}, where it is argued that
the computational difficulties can be overcome using modern
algorithms, while overfitting can be addressed by regularizing the
problem in (\ref{eq:IR}), i.e., fitting ``less complex'' isotonic
models than the optimal solution of  (\ref{eq:IR}). The Isotonic
Recursive Partitioning (IRP) algorithm proposed in
\citeasnoun{Spou2003} and \citeasnoun{Luss2012} (following previous
related work by \citeasnoun{Maxw1985} and \citeasnoun{Round1986},
among others) can easily solve problems with tens of thousands of
observations, and is based on recursive partitioning of the
covariate space and constructing isotonic models of increasing
complexity, thus generating a {\em regularization path} of isotonic
models in the sense that isotonic models along the path are
regularized by the number of partitions made.



In this paper, we focus on a more general form of isotonic regression
that minimizes a convex loss function subject to the isotonicity
constraints, i.e., we solve
\begin{equation}
\min{\{f(\hat{y}) : \hat{y}_i\leq \hat{y}_j\quad\forall(i,j)\in\mathcal{I}\}}
\label{eq:IC}\end{equation}
where $f: \reals^n\rightarrow\reals$ is a separable function such that
\begin{equation}
f(\hat{y})=\displaystyle\sum_{i=1}^n{f_i(\hat{y}_i)},
\label{eq:IC_obj}\end{equation} and $f_i: \reals\rightarrow\reals$
is differentiable and convex for all $i=1,\ldots,n$.  Typically,
$f_i(\hat{y}_i) = g(\hat{y}_i,y_i)$ measures the fit of $\hat{y_i}$
to the observed response. Table \ref{table:loss_functions} provides
several examples for the functions $f_i(\cdot)$ that define
$f(\cdot)$ above.

\begin{table}[h!]
\begin{center}
\extrarowheight 0.8ex
\begin{tabular}{|l|l|}
\hline
p-norm loss, $1<p<2$ & $(\hat{y}_i-y_i)^p$ \\ \hline
$\delta$-Huber loss &  $(\hat{y}_i-y_i)^2/2$ for $|\hat{y}_i-y_i|<\delta$ and $\delta(|\hat{y}_i-y_i|-\delta/2)$ otherwise\\ \hline
Negative Poisson log-likelihood & $\hat{y}_i-y_i\ln{(\hat{y}_i)}$\\ \hline
Negative Bernoulli log-likelihood & $-y_i\ln{(\hat{y}_i)}-(1-y_i)\ln{(1-\hat{y}_i)}$\\
\hline
\end{tabular}
\end{center}
\caption{Examples of loss functions solvable by Algorithm \ref{alg:GIRP}}
\label{table:loss_functions}
\end{table}

The notion of generalized isotonic regression is not new.  \citeasnoun{Barl1972} defined a generalized loss function as
\begin{equation} \label{eq:bb_objective}
f(\hat{y})= \displaystyle\sum_{i=1}^n{\Phi(\hat{y}_i)-\hat{y}_iy_i}
\end{equation}
for proper convex $\Phi:\reals\rightarrow\reals$ which was minimized
subject to the same isotonicity constraints as in (\ref{eq:IC}).
They showed that this generalized isotonic regression problem can be
solved equivalently as an instance of the $l_2$ isotonic regression
(\ref{eq:IR}).  This implies that any large-scale algorithm for
problem (\ref{eq:IR}) can be used to solve isotonic regressions with
objectives of the form (\ref{eq:bb_objective}).  Generalized
objectives for large-scale Poisson and Bernoulli regressions as
given in Table \ref{table:loss_functions} can be solved in this
manner, however the $p$-norm and Huber loss functions cannot.  This
relationship will be further formalized in Section
\ref{ss:barlow_brunk}.  Generalized isotonic regression (using
separable loss functions) in $d=1$ dimension was also considered in
\citeasnoun{Best2000} and \citeasnoun{Ahuja2001} using extensions of
the pooled adjacent violators algorithm (PAVA).  Neither assumes
differentiability as done here, and hence both are amenable to a
broader class of loss functions, albeit only in one dimension.
\citeasnoun{Hoch2003} offer a very efficient and related algorithm
for problem (\ref{eq:IC}) where the fits are restricted to being
integer (called the convex cost closure problem).  Their algorithm
can be extended to the continuous case in the sense of determining
an $\epsilon$-accurate solution by solving the integer problem on an
$\epsilon$-grid.  Depending on the required level of accuracy, their
approach can be computationally competitive with our algorithm as
described below for finding the solution of (\ref{eq:IC}), however
it lacks the natural statistical interpretation as a regularization
path which our approach affords. The method of \citeasnoun{Hoch2003}
is discussed in more detail in Section \ref{ss:barlow_brunk}.

The main contribution of this paper is a generalization of IRP that
can be used to solve large-scale multivariate generalized isotonic
regressions of the form (\ref{eq:IC}). Our generalization extends
the methods to any convex differentiable loss function, including
those mentioned in Table \ref{alg:GIRP}, and we term it generalized
isotonic recursive partitioning (GIRP). As with IRP, the
partitioning algorithm here addresses both of the main difficulties
with isotonic regression discussed above. Firstly, it provides a sequence of isotonic models of increasing complexity, converging to the globally optimal generalized isotonic solution. Early stopping along  this ``regularization path'' is a useful approach to overcome overfitting concerns of the globally optimal solution; less complex isotonic
models along the path often predict more accurately than the final
overfit model.  Secondly, it is computationally efficient; the
partitioning algorithm is an iterative scheme in which each
iteration partitions a group of observations by solving a structured
linear program for which very efficient algorithms exist.

It should be emphasized that, while the algorithmic modification from
IRP to GIRP is quite minor and the proofs for GIRP properties are
closely related to proofs of same results for IRP, the
generalization has important practical implications because it
significantly expands the range of applications for isotonic
modeling, as our examples below illustrate.

The paper continues as follows.  Section \ref{s:cirp} describes the
known results for $l_2$ isotonic regression and generalizes them to
the class of convex loss functions described by (\ref{eq:IC_obj}).
The generalized algorithm is described, and the relationship to
\citeasnoun{Barl1972} is formalized.  Section \ref{s:performance}
applies the results with Poisson log-likelihood and Huber's loss
functions to synthetic and real data sets.  A Matlab-based software
package implementing our results is available at
\href{http://www.tau.ac.il/~saharon/files/GIRPv1.zip}{\tt
www.tau.ac.il/\~{}saharon/files/GIRPv1.zip}.  We first define
terminology to be used throughout the remainder of the paper.

\subsection{Definitions}
Let $V=\{x_1,\ldots,x_n\}$ be the covariate vectors for $n$ training points where
$x_i\in\mathbb R^d$ and denote $y_i\in\mathbb R$ as the $i^{th}$ observed response.  We will refer to a general subset of
points $A\subseteq V$ with no holes (i.e., $x\preceq y\preceq z$ and $x,z\in A\Rightarrow y\in A$) as a \emph{group}.  Throughout the paper, we will use the shorthand $i\in A=\{i:x_i\in A\}$. Denote by $|A|$ the cardinality of group $A$.  The \emph{weight} of a group $A$ is denoted by
\begin{equation} \label{eq:weight}
w_A=\argmin_{z\in\reals}{\displaystyle\sum_{i\in A}{f_i(z)}}.
\end{equation}
For two groups $A$ and $B$, we denote $A\preceq B$ if $\exists x\in A, y\in B$ such that $x\preceq y$ and $\nexists x\in A, y\in B$ such that $y\prec x$ (i.e., there is at
least one comparable pair of points that satisfy the direction of
isotonicity).  A set of groups $\mathcal{V}$ is called isotonic if $A\preceq B\Rightarrow w_A\leq w_B, \forall A,B\in\mathcal{V}$.  A subset $\mathcal{L}$
($\mathcal{U}$) of $A$ is a \emph{lower set} (\emph{upper set}) of
$A$ if $x\in A,y\in\mathcal{L},x\prec y\Rightarrow x\in\mathcal{L}$
($x\in\mathcal{U},y\in A,x\prec y\Rightarrow y\in\mathcal{U}$). A group $B\subseteq A$ is defined as a block of group $A$ if
$w_{\mathcal{U}\cap B}\leq w_B$ for each upper
set $\mathcal{U}$ of $A$ such that $\mathcal{U}\cap B \ne \{\}$ (or
equivalently if  $w_{\mathcal{L}\cap B}\geq
w_B$ for each lower set $\mathcal{L}$ of $A$ such that
$\mathcal{L}\cap B \ne \{\}$). We denote the optimal solution for minimizing $f(y)$ in the variable $y$ by $y^*$, i.e., $y^*=\argmin{f(y)}$.  The notation $(\partial f(\hat{y})/\partial \hat{y})|_{y}$ denotes the derivative of a function $f$ with respect to the variable $\hat{y}$ evaluated at the point $y$.

\section{Generalized Isotonic Recursive Partitioning with Convex Loss Functions} \label{s:cirp}
In this section, we generalize the results for $l_2$ isotonic partitioning resulting in an algorithm termed Generalized Isotonic Recursive Partitioning (GIRP), and derive useful properties of this generalization.  The solution at each iteration, as in IRP, is defined by groups that are proven to be the union of blocks in the optimal solution.  Section \ref{ss:irp} first gives an overview of the IRP algorithm.  Section \ref{ss:splitting_algo} then details the partitioning step for the generalized case and derives the resulting GIRP algorithm.  When $f_i(\hat{y}_i)=(\hat{y}_i-y_i)^2$, that is, $f(\cdot)$ is the $l_2$ loss function, all results in this section replicate those of IRP and $w_A$ becomes the average of the observations in group $A$.  Section \ref{ss:path_conv} proves convergence of the partitioning algorithm to the global optimal solution of (\ref{eq:IC}) and shows that the solution at each iteration of the algorithm is isotonic, i.e., the iterations provide a regularization path of isotonic solutions.

\subsection{Isotonic Recursive Partitioning with the $l_2$ Loss Function} \label{ss:irp}
We here briefly review the ideas of the IRP algorithm of \citeasnoun{Spou2003} and \citeasnoun{Luss2012}.  The optimal solution to the $l_2$ isotonic regression problem (\ref{eq:IR}) is known to be defined by a partitioning of the observations $y_i$ into blocks in which $\hat{y}_i=\hat{y}_j$ if observations $y_i$ and $y_j$ are in the same block.  Indeed, this structure can be seen through the optimality conditions (i.e., Karush-Kuhn-Tucker (KKT) conditions, see \citeasnoun{Boyd2004}) for problem (\ref{eq:IR}).  The optimal solution to (\ref{eq:IR}), denoted by $\hat{y}^*$, satisfies these conditions, which are given by
\begin{enumerate}[label=(\alph*)]
\item $2(\hat{y}^*_i-y_i)=\sum_{j:(j,i)\in\mathcal{I}}{\lambda^*_{ji}}-\sum_{j:(i,j)\in\mathcal{I}}{\lambda^*_{ij}}
\quad\forall i\in V$
\item $\hat{y}^*_i\leq\hat{y}^*_j\quad\forall(i,j)\in\mathcal{I}$
\item $\lambda^*_{ij}\geq0\quad\forall (i,j)\in\mathcal{I}$
\item $\lambda^*_{ij}(\hat{y}^*_i-\hat{y}^*_j)=0\quad\forall
(i,j)\in\mathcal{I}$,
\end{enumerate}
where $\lambda_{ij}^*$ is the optimal dual variable associated with isotonicity constraint $\hat{y_i}\leq\hat{y_j}$.  Convexity of the $l_2$ loss function implies that any solution satisfying conditions (a)-(d) is a globally optimal solution. From condition (d), $\lambda_{ij}^*>0\Rightarrow\hat{y}^*_i=\hat{y}_j^*$, i.e., the optimal solution is made of blocks $V$ where $\lambda_{ij}^*>0$ for all $i,j\in V$, which implies that all observations within a block $V$ are fit to the same value.  Furthermore, when restricting all fits within a block $V$ to be equivalent, the isotonic regression problem over block $V$ is the following unconstrained quadratic program
\begin{equation} \label{eq:irp_weight}
\min_{\hat{y}\in\reals}{\{\displaystyle\sum_{i\in V}{(\hat{y}-y_i)^2}\}}
\end{equation}
which is trivially solved at $\hat{y}^*=\overline{y}_V$ where $\overline{y}_V$ denotes the average of all observations in block V. Condition (b) implies that these averages must satisfy isotonicity, i.e., if $V^-\preceq V^+$ are isotonic blocks then $\overline{y}_{V^-}\leq \overline{y}_{V^+}$.  Thus, the structure of the optimal solution is a partitioning of the set $\{1,\ldots, n\}$ into some (unknown number) $K$ blocks $\{V_1,\ldots,V_K\}$ where $\hat{y}_i^*=\overline{y}_{V_k}$ for all $i\in V_k$ and $\hat{y}^*_i\leq\hat{y}^*_j$ for all $(i,j)\in\mathcal{I}$.

Many such feasible partitions exist that are not optimal (e.g., set all fits to the average of all observations).  Condition (a) above must also be satisfied and gives the motivation for how to partition the set of observations in a manner that leads to the optimal partitioning.  The partitioning scheme is detailed for the general case of convex loss functions in the next subsection.  Here, we only give a general idea of how IRP works.

IRP starts with the entire dataset as one group $V$ and iteratively splits it into an increasing number of groups, until the optimal solution of (\ref{eq:IR}) is reached. At each iteration, the algorithm chooses a sub-optimal group and partitions it into two groups by solving a specially structured linear program, detailed in the next subsection, that is amenable to very efficient algorithms.  If the partition puts all observations into one group, it can be shown that the group is a block, i.e., optimal.  Otherwise, the fits in the two resulting groups are recalculated as their averages (via (\ref{eq:irp_weight}) above), while the rest of the groups and their fits remain at their values in the previous iteration.  IRP is thus an iterative scheme that splits a group at each iteration and never merges two groups back together; therefore, IRP is referred to as a no-regret partitioning algorithm.


Two important theorems are proven \cite{Luss2012} with respect to IRP. The first states that the new solution obtained after each partitioning step still satisfies isotonicity.  After iteration $K$, there are $K+1$ groups, $V_1,\ldots,V_{K+1}$, in the partitioning with $\hat{y}_i$ fit to $\overline{y}_{V_k}$ for all $i\in V_k$ with $k\in\{1,\ldots,K+1\}$.  The theorem thus says, that at each iteration $K$, the fits from the partitioning $V_1,\ldots,V_{K+1}$ provide a potential isotonic prediction model.   The second theorem shows that the IRP scheme terminates at the globally optimal partitioning.  Hence, IRP produces a path of increasingly complex (since each iteration adds a partition) isotonic solutions, terminating in the optimal solution of (\ref{eq:IR}).  These theorems are made possible because of the particular splitting criterion used in the IRP algorithm, which is amenable to efficient calculation as mentioned above. The generalized version of the splitting criterion and the resulting algorithm are discussed next.

\subsection{The partitioning algorithm} \label{ss:splitting_algo}
As with IRP, we solve a sequence of subproblems in order to solve the generalized isotonic regression problem (\ref{eq:IC}); each subproblem divides a group of observations into two groups at each iteration.  An important property of IRP with the $l_2$ loss function is that observations separated at one iteration remain separated at all future iterations.  The same property applies here and implies that the total number of iterations is bounded by the number of observations $n$.

The partitioning algorithm is motivated by the optimality conditions for the generalized isotonic regression problem (\ref{eq:IC}).  The optimal solution to (\ref{eq:IC}), denoted by $\hat{y}^*$, are identical to conditions (a)-(d) in Section \ref{ss:irp} above, with the exception that condition (a) now has the generalized form
\begin{enumerate}[label=(\alph*)]
\item
$\frac{\partial f_i(\hat{y}_i)}{\partial \hat{y}_i}\big|_{\hat{y}_i^*}=\sum_{j:(j,i)\in\mathcal{I}}{\lambda^*_{ji}}-\sum_{j:(i,j)\in\mathcal{I}}{\lambda^*_{ij}}\quad\forall i\in V$
\end{enumerate}
where again $\lambda_{ij}^*$ is the optimal dual variable associated with isotonicity constraint $\hat{y_i}\leq\hat{y_j}$.  Convexity of the loss function again implies that any solution satisfying the optimality conditions is a globally optimal solution.  The structure of the optimal solution as a partitioning of isotonic blocks can be seen from the KKT conditions as described in Section \ref{ss:irp}.
Within each block, the fit to each observation for the general case is taken to be the weight of the observations in the block as defined by (\ref{eq:weight}).  Isotonicity of the two blocks $V^-$ and $V^+$, i.e., $V^-\preceq V^+$, means that $w_{V^-}\leq w_{V^+}$.  From condition (a), summing over all observations in a block $V$, i.e., optimal group, gives
\begin{equation} \label{eq:partition_criterion}
\displaystyle\sum_{i\in V}{\frac{\partial f_i(\hat{y}_i)}{\partial \hat{y}_i}}\bigg|_{\hat{y}_i^*}=0.
\end{equation}

Derivation of the partitioning step is as follows.  Consider a group $V$ where $\hat{y}_i^*=w_V$ for all $i\in V$.   If $V$ is an optimal group, it is a block and must satisfy (\ref{eq:partition_criterion}).  If it is not optimal, however, we can find a partitioning of $V$ into two isotonic groups $V^+$ and $V^-$ such that
\begin{equation} \label{eq:partition_criterion_ineq}
\sum_{i\in V^+}
{\frac{\partial f_i(\hat{y}_i)}{\partial \hat{y}_i}\bigg|_{w_V}}-\sum_{i\in V^-}{\frac{\partial f_i(\hat{y}_i)}{\partial \hat{y}_i}\bigg|_{w_V}}<0.
\end{equation}
The first summation over $i\in V^+$ is the change in the objective value of problem (\ref{eq:IC}) due to an increase in the fits of observations in $V^+$.  The second summation over $i\in V^-$ is the change in the objective value due to a decrease in the fits of observations in $V^-$.  Such a partition thus means that increasing the fits in $V^+$ to be greater than $w_V$ while decreasing the fits in $V^-$ to be less than $w_V$ (which by definition maintains isotonicity of the fits) will cause an overall decrease in the objective value to problem (\ref{eq:IC}).  Fits that decrease the overall objective value can be achieved by fitting the observations in $V^+$ and $V^-$ to their respective weights $w_{V^+}$ and $w_{V^-}$.  Hence, we search for an isotonic partitioning of $V$ into $V^+$ and $V^-$ that minimizes the lefthand term in (\ref{eq:partition_criterion_ineq}).

Denote by $\mathcal{C}_V=\{(V^-,V^+):V^-,V^+\subseteq V,V^-\cup V^+=V,V^-\cap V^+=\{\},\not\exists x\in V^-,y\in V^+\quad\mbox{s.t.}\quad y\preceq x\}$
the set of all feasible (i.e., isotonic) partitions defined by observations in $V$.  Partitioning is referred to as making a cut through the variable space (hence the optimal partition is made by an \emph{optimal cut}).  The optimal cut is determined as
the partition that solves the problem
\begin{equation}
\min_{(V^-,V^+)\in\mathcal{C}_V}{\{\sum_{i\in V^+}
{\frac{\partial f_i(\hat{y}_i)}{\partial \hat{y}_i}\bigg|_{w_V}}-\sum_{i\in V^-}{\frac{\partial f_i(\hat{y}_i)}{\partial \hat{y}_i}\bigg|_{w_V}}\}}
\label{eq:optimal_cut}
\end{equation}
where $V^-$($V^+$) is the group on the lower (upper) side of the edges of the cut.  The optimal cut
problem (\ref{eq:optimal_cut}) can be expressed as the binary
program
\begin{equation}
\min{\{\displaystyle\sum_{i\in V}{x_i\frac{\partial f_i(\hat{y}_i)}{\partial\hat{y}_i}\bigg|_{w_V}} : x_i\leq x_j\quad\forall (i,j)\in\mathcal{I}, x_i\in\{-1,+1\}\quad\forall i\in V\}}.  \label{eq:optimal_cut_bp}\end{equation}
It is well-known \cite{Murt1983} that the continuous relaxation to this binary program (i.e., replacing the constraints $x_i\in\{-1,+1\}$ by $-1\leq x_i\leq1$ for all $i\in V$) is solved on the boundary of the feasible region with $x_i^*\in\{-1,+1\}$ for all $i=1\ldots n$.  Thus the optimal cut problem (\ref{eq:optimal_cut}) is equivalent to solving the linear program
\begin{equation}
\min{\{z^Tx : x_i\leq x_j\quad\forall (i,j)\in\mathcal{I}, -1\leq x_i\leq1\quad\forall i\in V\}}  \label{eq:optimal_cut_lp}\end{equation}
where $z_i=(\partial f_i(\hat{y}_i)/\partial\hat{y}_i)|_{w_V}$.  Problem (\ref{eq:optimal_cut_lp}) with $z_i=2(\overline{y}_V-y_i)$ gives the linear program used to make partitions in IRP with $l_2$ loss function as described above in Section \ref{ss:irp}.  As seen by property (\ref{eq:partition_criterion_ineq}), a property of this optimal cut for generalized isotonic regression is that the sum of loss functions with $x_i=+1$ ($x_i=-1$) can be decreased by increasing (decreasing) the corresponding fits.  That is, by increased (decreasing) $w_v$ for observations $i$ with $x_i=+1$ ($x_i=-1$), the total change in loss is decreased, i.e.,
\begin{equation} \label{eq:cut_properties}
\sum_{\{i:x_i=+1\}}{\frac{\partial f_i(\hat{y}_i)}{\partial\hat{y}_i}\bigg|_{w_V}}\leq0\quad\mbox{and}\quad\sum_{\{i:x_i=-1\}}{\frac{\partial f_i(\hat{y}_i)}{\partial\hat{y}_i}\bigg|_{w_V}}\geq0.
\end{equation}

This group-wise partitioning operation is the basis for our algorithm which is detailed in Algorithm \ref{alg:GIRP}.  The algorithm differs from IRP only in Step 6 (they are obviously identical when $f(\cdot)$ is the $l_2$ loss).  Initially, all observations are in one group.  Each iteration splits a group optimally by solving subproblem~(\ref{eq:optimal_cut_lp}). A list $\mC$ of potential optimal cuts for each group generated thus far is maintained, and, at each iteration, the cut among them with the smallest (most negative) objective value is performed. Partitioning of a group ends when the solution to~(\ref{eq:optimal_cut_lp}) is trivial (i.e., no split is found
because the group is a block).  The algorithm stops when no further groups can be partitioned.

\begin{algorithm}
\caption{Generalized Isotonic Recursive Partitioning}
\begin{algorithmic} [1]
\REQUIRE Observations $(x_1,y_1),\ldots,(x_n,y_n)$ and partial order
$\mathcal{I}$.
\REQUIRE $k=0,\mathcal{A}=\{\{x_1,\ldots,x_n\}\}$,$\mathcal{C}=\{(0,\{x_1,\ldots,x_n\},\{\})\}$,$\mathcal{B}=\{\}$, $M_0=(\mathcal{A},w_\mathcal{A})$.
\WHILE{$\mathcal{A}\ne\{\}$}
\STATE Let $(val,w^-,w^+)\in\mathcal{C}$ be the potential partition with largest $val$.
\STATE Update $\mathcal{A}=(\mathcal{A}\setminus (w^-\cup w^+))\cup\{w^-,w^+\}$, $\mathcal{C}=\mathcal{C}\setminus(val,w^-,w^+)$.
\STATE $M_k=(\mathcal{A},\overline{y}_{\mathcal{A}})$.
\FORALL{$v\in\{w^-,w^+\}\setminus\{\}$}
\STATE Set $c_i=\frac{\partial f_i(\hat{y}_i)}{\partial\hat{y}_i}\big|_{w_v}$ $\forall i\in v$ where $w_v$ is the weight (\ref{eq:weight}) of the observations in $v$.
\STATE Solve LP (\ref{eq:optimal_cut_lp}) with input $z$ and get $z^*=\argmin{\mbox{LP}(\ref{eq:optimal_cut_lp})}$.
\IF{$z_1^*=\ldots=z_n^*$ (group is optimally divided)}
\STATE Update $\mathcal{A}=\mathcal{A}\setminus v$ and $\mathcal{B}=\mathcal{B}\cup \{v\}$. \ELSE
\STATE Let $v^-=\{x_i:z^*_i=-1\}, v^+=\{x_i:z^*_i=+1\}$.
\STATE Update $\mathcal{C}=\mathcal{C}\cup\{(c^Tz^*,v^-,v^+)\}$
\ENDIF
\ENDFOR
\STATE k=k+1.
\ENDWHILE
\RETURN $\mathcal{M}$, a sequence of isotonic models, where $M_k$ contains the $k^{th}$ iteration's partitioning of observations and corresponding group weights.
\end{algorithmic}
\label{alg:GIRP}
\end{algorithm}

Each iteration $k$ of Algorithm \ref{alg:GIRP} produces a model $M_k$ by fitting each group in $M_k$ to its weight. For a set of groups $\mathcal{V}=\{V_1,\ldots,V_k\}$, denote $w_\mathcal{V}=\{w_{V_1},\ldots,w_{V_k}\}$.  Then model $M_k=(\mathcal{V},w_\mathcal{V})$ contains the partitioning $\mathcal{V}$ as well as a fit to each of the observations, which is the weight, as defined by (\ref{eq:weight}), of the group it belongs to in the partition.

\subsection{Properties of the partitioning algorithm} \label{ss:path_conv}
So far, we have detailed the partitioning algorithm which is based on iteratively solving problem (\ref{eq:optimal_cut_lp}), but we have not yet shown that partitioning according to this particular scheme, i.e., solving problem (\ref{eq:optimal_cut_lp}), optimally solves the generalized isotonic regression problem. Theorem \ref{th:no_regret_cut} next states the main result that implies Algorithm \ref{alg:GIRP} is a no-regret partitioning algorithm for (\ref{eq:IC}) (no-regret in the same sense as described for IRP in Section \ref{ss:irp}).  In the case of $l_2$ isotonic regression, this result is already known \cite{Maxw1985,Spou2003,Luss2012}.  This theorem leads to our convergence result.  The proof requires straightforward changes to the proof in \citeasnoun{Luss2012} based on the definition of convexity, the new algorithm cut in (\ref{eq:optimal_cut_lp}), and its properties (\ref{eq:cut_properties}); the proof is thus left to the Appendix.

\begin{theorem} \label{th:no_regret_cut}
Assume group $V$ is the union of blocks from the
optimal solution to problem (\ref{eq:IC}). Then a cut made by
solving (\ref{eq:optimal_cut_lp}) at a particular iteration does not cut through any block in the global optimal solution.
\end{theorem}

The case of multiple observations at the same coordinates can be disregarded.  Let $\mathcal{V}$ denote a set of groups where each group in $\mathcal{V}$ contains observations with the same coordinates, i.e $V\in\mathcal{V}$ denotes the indices of multiple observations and $|V|=1$ means that $V$ is a single observation.  Then, we define $g_V(\hat{y}_V)=\sum_{i\in V}{f_i(\hat{y}_V)}$ where $\hat{y}_V\in\reals$ and modify $f(\cdot)$ in the generalized isotonic regression problem (\ref{eq:IC}) to be
$f(\hat{y})=\sum_{V\in\mathcal{V}}{g_V(\hat{y}_V)}$ where $\hat{y}\in\reals^{|\mathcal{V}|}$ and each function $g_V(\cdot)$ satisfies the necessary properties for applying GIRP.

Since Algorithm~\ref{alg:GIRP} starts with the union of all blocks for the first partition, we can conclude from this theorem that Algorithm~\ref{alg:GIRP} never cuts a block when generating partitions.  From the derivation of the partitioning problem, it is clear that if an isotonic partition can be made, it will be made; that is, the algorithm will not stop early.  Convergence of Algorithm \ref{alg:GIRP} to the global isotonic solution with no regret then follows by repeatedly applying Theorem \ref{th:no_regret_cut} until all blocks of the optimal solution are identified.  The next theorem states that Algorithm \ref{alg:GIRP} provides isotonic solutions at each iteration.  This result implies that the path of solutions generated by Algorithm \ref{alg:GIRP} can be regarded as a regularization path for the generalized isotonic regression problem (\ref{eq:IC}).  Proof of this theorem is again held until the Appendix for the same reasons given above.
\begin{theorem} \label{th:isotonic_solutions}
Model $M_k$ generated after iteration $k$ of Algorithm \ref{alg:GIRP} is in the class of isotonic models.
\end{theorem}

Complexity analysis of Algorithm \ref{alg:GIRP} depends on the number of observations $n$ and isotonic constraints $m$, and the complexity of solving linear program (\ref{eq:optimal_cut_lp}).  Firstly, we assume that computing the weight of a group $V$ via (\ref{eq:weight}) requires computationally less effort than solving problem (\ref{eq:optimal_cut_lp}) (in practice these problems are one-dimensional convex minimization problems that are easily solved with a binary search). In short, linear program (\ref{eq:optimal_cut_lp}) is dual to a linear maximum flow network problem \cite{networkflows}, which is a well-studied problem.  It can be solved in $O(mn\log{n})$ \cite{Slea1983} or $O(n^3)$ \cite{Gali1980} in the general case that we consider; special cases such as $n=2$ or where the observations lie on a grid can be computed even faster \cite{Spou2003}.  Choice of algorithm depends on $m$ which is $O(n^2)$ in the worst case.   Given GIRP requires at most $n$ iterations, this leads to worst case complexities of $O(mn^2\log{n})$ or $O(n^4)$.  A recent problem reduction by \citeasnoun{Stou2010} can be used to obtain an equivalent representation of the desired problem with $d$-dimensional data and $O(n\log^{d-1}{n})$ constraints and observations, which can be useful when $m$ is large. Finally, \citeasnoun{Luss2012} show that IRP performs in $O(Cn^3)$ in practice, where $C$ is a function of the fraction of observations on each of the cut at each iteration.  The same result applies here.

\subsection{Relations to Other Generalized Isotonic Regressions} \label{ss:barlow_brunk}
We here formalize the relationship between GIRP and the work of \citeasnoun{Barl1972}, which was hinted at in \citeasnoun{Luss2012} and mentioned in the introduction above.  The generalized isotonic regression problem of \citeasnoun{Barl1972} is of the form
\begin{equation} \label{eq:IR_gen_barlow}
\min{\{\displaystyle\sum_{i=1}^n{\Phi(\hat{y}_i)-\hat{y}_iy_i} : \hat{y}_i\leq\hat{y}_j\quad\forall(i,j)\in\mathcal{I}\}}
\end{equation}
in $\hat{y}\in\reals^n$ where we have left out weights for simplicity.  While they allow $\Phi : \reals\rightarrow\reals$ to be nondifferentiable, we consider here only the differentiable case and denote $\phi(\cdot)$ as the derivative of $\Phi(\cdot)$.  Let $\hat{z}^*$ be the solution of (\ref{eq:IR}) (i.e., $l_2$ isotonic regression) with given observations $y_i$.  Theorem 3.1 of \citeasnoun{Barl1972} claims that the solution to (\ref{eq:IR_gen_barlow}) can be obtained as
\begin{equation} \label{eq:barlow_transform}
\hat{y}_i^*=\phi^{-1}(\hat{z}_i^*)
\end{equation}
where $\phi^{-1}(\cdot)$ is defined by $\phi^{-1}(\phi(x))=x$ for all $x\in\reals$.  Thus, any objective of the form (\ref{eq:IR_gen_barlow}) can be solved by computing the $l_2$ isotonic regression on input observations and then transforming the solution using (\ref{eq:barlow_transform}). In this manner, IRP can be used to solved the somewhat limited class of generalized isotonic regression problems defined in \citeasnoun{Barl1972} (note that without requiring $\Phi(\cdot)$ proper, their theory would apply to any convex loss function). However, it is also clear that Algorithm \ref{alg:GIRP} provides the tools for solving more general isotonic regression problems than (\ref{eq:IR_gen_barlow}), e.g., as in the case for the $p$-norm or Huber's loss function.

The same transformation from \citeasnoun{Barl1972} can be used to derive an isotonic regularization path for generalized isotonic regression problems with the structure of (\ref{eq:IR_gen_barlow}). Indeed, this can be shown using the above framework for our general isotonic regression problem (\ref{eq:IC}), and is formalized in Proposition \ref{prop:barlow_girp}.

\begin{proposition} \label{prop:barlow_girp}
Problem (\ref{eq:IR_gen_barlow}) can be solved either by
\begin{enumerate}
\item Applying IRP to the observation data $y$ to obtain $\hat{z}^*$ and tranforming using (\ref{eq:barlow_transform}),
\item Applying GIRP directly to (\ref{eq:IR_gen_barlow}).
\end{enumerate}
Furthermore, both algorithms are equivalent when applied to (\ref{eq:IR_gen_barlow}) in the sense that the regularization path of partitions for each algorithm are equivalent.
\end{proposition}
\begin{proof}
IRP can be used to solve the $l_2$ isotonic regression problem to obtain $\hat{z}^*$.  Application of Theorem 3.1 of \citeasnoun{Barl1972} gives the solution to (\ref{eq:IR_gen_barlow}) via (\ref{eq:barlow_transform}).  In order to apply GIRP, let $f_i(\hat{y}_i)=\Phi(\hat{y}_i)-\hat{y}_iy_i$ and $(\partial f_i(\hat{y}_i)/\partial\hat{y}_i)|_{\hat{y}_i}=\phi(\hat{y}_i)-y_i$ where $\phi(\cdot)$ denotes the derivative of $\Phi(\cdot)$. Then $w_V$ as defined by (\ref{eq:weight}) satisfies $\sum_i{(\phi(w_V)}-y_i)=0$ giving $w_V=\phi^{-1}(\overline{y}_V)$ where $\overline{y}_V$ denotes the mean observation over group $V$ and $\phi^{-1}(\cdot)$ is defined by $\phi^{-1}(\phi(x))=x$ for all $x\in\reals$.  Hence, here $(\partial f_i(\hat{y}_i)/\partial\hat{y}_i)|_{w_V}=\overline{y}_V-y_i$ and the GIRP partitioning problem (\ref{eq:optimal_cut_lp}) is equivalent to the corresponding partition problem in \citeasnoun{Luss2012}.
\end{proof}

The only difference between IRP and GIRP for solving problem (\ref{eq:IR_gen_barlow}) is that GIRP fits observations to the transformed isotonic regression fits $w_V=\phi^{-1}(\overline{y}_V)$ along the path, while IRP fits observations to the mean of the group's observations and the transformation is done on the final optimal partitioning.  It is easy to see that the transformation (\ref{eq:barlow_transform}) can be applied to each iteration of IRP along the path in order to obtain an equivalent path to that of GIRP.

This connection also relates the algorithm of \citeasnoun{Maxw1985}, which solves problem $(\ref{eq:IR_gen_barlow})$ with $\Phi(\hat{y}_i)=1/\hat{y}_i$, to IRP.  Due to the analysis here, these algorithms are actually equivalent.  Both the problem of \citeasnoun{Maxw1985} and $l_2$ isotonic regression are specific instances of the more general problem (\ref{eq:IC}) solved in this paper.  It should be noted that \citeasnoun{Maxw1985} did not make use of, or even recognize, the regularization path which plays a significant role for isotonic regression in dimension $d>1$.

Lastly, \citeasnoun{Hoch2003} offer another partitioning algorithm for problem (\ref{eq:IC}) with additional integer constraints.   GIRP, in the continuous case, solves cut problem (\ref{eq:optimal_cut}) because we know the fit within optimal groups (i.e., the weight of the group).  Rather, in the integer case, the cut problem (\ref{eq:optimal_cut}) is solved instead with the derivatives evaluated at some $\alpha$ taken as the median of an interval in which the optimal fits lie.  A theorem states that this partition problem divides the group into two groups $V^-$ and $V^+$, where optimal fits to observations in $V^-$ are less than $\alpha$ and optimal fits to observations in $V^+$ are greater than $\alpha$.  The problem is thus stated as determining a sequence $\alpha_1,\ldots,\alpha_l$ such that observations with optimal fits in the interval $(\alpha_i,\alpha_{i+1}]$ have the same optimal fit.  Given a criterion for determining when $\alpha$ is a breakpoint in this sequence, their algorithm can do better than a binary search.  In fact, they further suggest a method that has a worst-case complexity equivalent to solving three max-flow problems. The complexity comes from using the information in previous max-flow problems to start new max-flow problems.  A similar idea could possibly be applied in our continuous case where the search for breakpoints uses the group weight in the cut problem. This highly efficient algorithm does not provide the exact solution to the continuous case, but a regularization path based on the bounds they get when searching for breakpoints can be considered for the integer case, and in turn, for the problem on an $\epsilon$-grid.

\subsection{Regularization by recursive partitioning}
\label{ss:reg_by_partitioning}
GIRP obtains the solution to problem (\ref{eq:IC}) by recursively partitioning the covariate space into progressively smaller regions and fitting the best constant in each region, referred to here as the weight which is defined by (\ref{eq:weight}). As such, it is natural to think of the resulting sequence of models created from early stopping as a regularization path of models of increasing complexity, indexed by the number of iterations of the algorithm.  Other examples of using early stopping for regularization include training neural networks with back propagation \cite{Caru2000} and boosting \cite{Ross2004}.  Extensive experience of the usefulness of regularization in high dimensional fitting \cite{Wahb1990,Tibs1996,Schol2001}, and especially in nonparametric models like isotonic regression, suggests that regularization, embodied in this case by early stopping of the algorithm, can lead to reduced overfitting and hence improved predictive performance. As Theorem \ref{th:isotonic_solutions} indicates, when stopping early and fitting the weight to each region, we are guaranteed to obtain a feasible isotonic model.

While GIRP uses early stopping for regularization of the globally optimal isotonic model, we note that regularization commonly refers to learning a model by explicitly constraining the family of models that are considered, and optimizing over this family. Early stopping after the $k^{th}$ iteration of GIRP produces an isotonic model with $k$ cuts obtained through a sequence of local optimization problems.  However, this model is not the solution to any global optimization problem.  The $k^{th}$ model of GIRP is thus only one potential model with $k$ cuts that satisfies the isotonicity constraints.  One might, for example, seek a regularized isotonic model that minimizes loss subject to the isotonicity constraints such that exactly $k$ cuts are made. The $k^{th}$ model in this case has a clear interpretation and more flexibility than the $k^{th}$ GIRP model.  While this would certainly be an interesting problem to consider, it is combinatorially difficult and the authors do not know of any efficient methods for solving it.

In \citeasnoun{Luss2012}, model complexity of $l_2$ isotonic regression along the IRP path is quantified through the concept of {\em equivalent degrees of freedom} (DFs) as defined by \citeasnoun{Efron1986} and \citeasnoun{Hastie2001}. The initial iterations of IRP are shown to perform much more fitting than later iterations, and this phenomenon becomes more pronounced as the dimension $d$ increases. For example, in dimension $6$, often $50\%$ of DFs were fitted by the first IRP iteration. Although the model complexity and DF measures of \citeasnoun{Efron1986} do not generalize to non-$l_2$ loss as used in GIRP, the general spirit of this result should persist. Intuitively, because the space of isotonic splits of the entire covariate space probed in the first iteration is much larger than the space of possible isotonic cuts in further iterations, finding the optimal first split corresponds to a significant portion of all fitting.

These two effects ---  importance of early stopping coupled with the high portion of fitting in earlier iterations ---  are demonstrated
empirically in the experiments of the next section, where the best
performing solution along the GIRP path is compared to the optimal
solution of problem (\ref{eq:IC}) in terms of predictive
performance.

\section{Performance evaluation} \label{s:performance}
We here demonstrate usefulness of the partitioning algorithm for generalized loss functions.  The contribution of our generalization is specifically illustrated by the use of Huber loss, which proves to be very effective in the case of outliers.  We first exhibit the computational performance of GIRP and show that the algorithm can be applied to large-scale problems.  We then consider synthetic data sets that demonstrate the impact of regularization and conclude with an example on real data.

\subsection{Practical Computational Performance} \label{ss:comp_performance}
Solving the multivariate isotonic regression problem with general loss functions such as Huber's loss was previously a computationally difficult problem.  For certain loss functions, the isotonic regression problem can be reformulated and solved with off-the-shelf convex optimization solvers.  For example, isotonic regression with Huber's loss can be reformulated as a quadratic program by adding many variables to the optimization problem.  Simulations with 1000 training points were solved in 2.3 seconds with GIRP versus 135 seconds using Mosek \cite{Mosek} to solve the quadratic program (averaged over 50 simulations).  This simple experiment demonstrates that GIRP, which is specifically designed for isotonic regression problems, is clearly a much more practical tool than using off-the-shelf generic solvers and makes generalized isotonic regression problems amenable to large-scale problems.

Figure \ref{fig:timetest} (left) illustrates that GIRP can solve large-scale problems with Huber's loss.  The $i^{th}$ observation in each simulation is generated as $y_i=(\prod_j{x_{ij}})+\mathcal{N}(0,d^2)$ with $x_{ij}\sim\mathcal{U}[0,2]$, $d$ representing the dimension, and outliers randomly inserted. Results are averages over 50 simulations.  Isotonic regression in 8 dimensions with 20,000 training instances is solved in less than one minute.  Figure \ref{fig:timetest} (right) shows the number of partitions that GIRP performs on average for varying dimensions.  More training data and higher dimension typically implies more complex isotonic models, resulting in more partitioning problems and more computational time.  The computational limitation of training the isotonic model with GIRP is solving the partitioning problem.  \citeasnoun{Luss2010b} further offers a heuristic for solving the partition problem that makes training isotonic regression problems with up to 200,000 training instances easily feasible.

\begin{figure}[h!] \begin{center}
  \begin{tabular} {cc}
     \psfrag{title}[b]{\small{Time vs \# Training Points}}
     \psfrag{sec}[b]{\small{Time (seconds)}}
     \psfrag{num}[t]{\small{Number Training Points}}
     \includegraphics[width=.45 \textwidth]{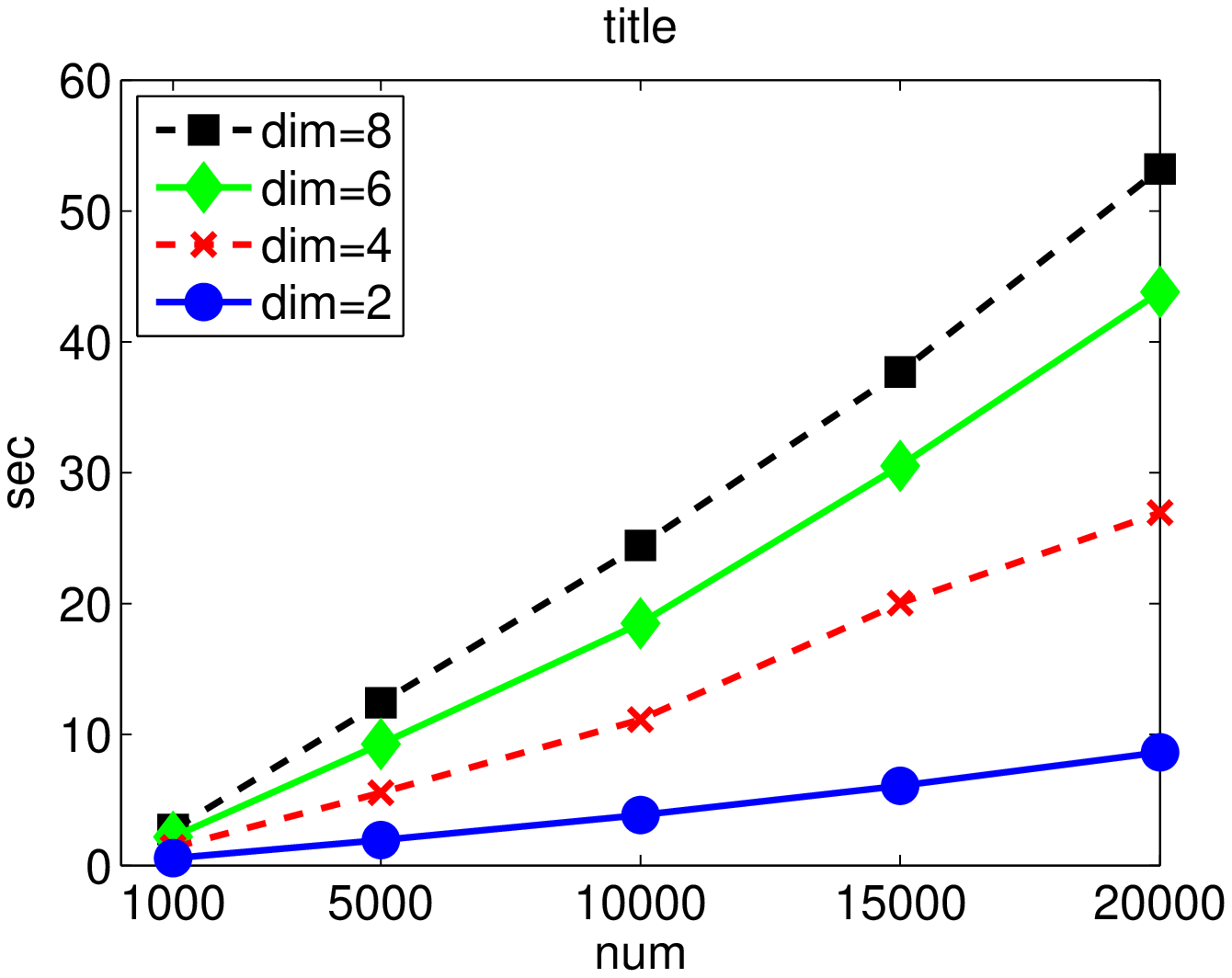} &
     \psfrag{title}[b]{\small{Number Partitions vs \# Training Points}}
     \psfrag{par}[b]{\small{Number of Partitions}}
     \psfrag{num}[t]{\small{Number Training Points}}
     \includegraphics[width=.45 \textwidth]{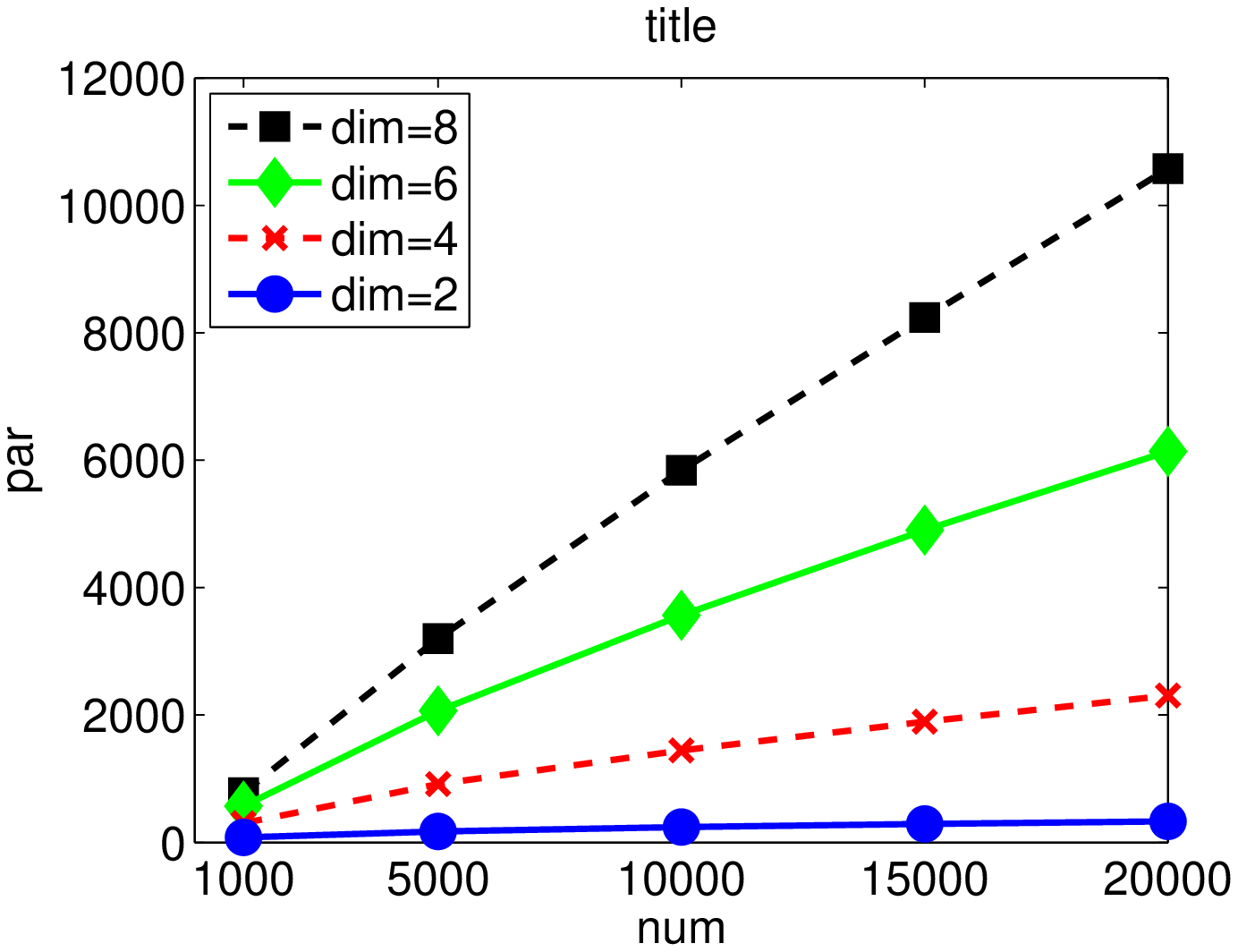}

     \end{tabular}
\caption{Left: Computational performance of training isotonic models with GIRP on a simulation with Huber's loss for varying number of dimensions and training instances. Right: Complexity of isotonic models as measured by the number of partitions for varying  number of dimensions and training instances.}
\label{fig:timetest}
\end{center} \end{figure}

\subsection{Simulations} \label{ss:perf_simulations}
Experiments are run on two different loss functions.  In the first experiment, count data is simulated from Poisson distributions where the average number of occurrences is generated by two different isotonic models. Generalized isotonic regression models for the Poisson rate are obtained by minimizing negative Poisson log-likelihood subject to isotonicity constraints.   In the second set of experiments, observations are generated by two different isotonic models and .5\% of the training observations are multiplied by a large constant to make them outliers.  Generalized isotonic regression models are obtained using $\delta$-Huber loss. Note that Poisson isotonic regressions can be handled using IRP due to the theory of \citeasnoun{Barl1972}, while Huber isotonic regressions require using GIRP.

Our experimental framework is as follows.   A training and testing set are independently simulated by a fixed distribution.  Training and testing sets have 15000 and 3000 observations, respectively.  A model is first generated on the training data.  In the case of GIRP, the training data is split into a subtraining set of 12000 observations and a validation set of the remaining 3000 observations.  A path of isotonic models is generated by running GIRP on the subtraining data.  The validation data is used to select the regularization level (stopping point), and the resulting model is applied to predict the testing data.  With respect to parametric regression, e.g., Poisson and Huber regressions, models are trained on the full 15000 observation training set and tested on the 3000 observation testing set.  Results are based on averaging fifty simulations.

The first two examples use Poisson negative log-likelihood as the loss function.  Data for the two simulations is generated as $x_{ij}\sim\mathcal{U}[0,10]$ and $x_{ij}\sim\mathcal{U}[5,10]$ (the coordinates of $x$ are drawn i.i.d in all our experiments), respectively.  The $i^{th}$ observation in each simulation is generated as $y_i\sim\mbox{Poisson}(\prod_j{\sqrt{x_{ij}}})$ and $y_i\sim\mbox{Poisson}(\sum_j{x^2_{ij}})$, respectively.  The isotonic models are compared to the results of a Poisson regression, and performance here is measured by negative Poisson log-likelihood.  The regularized model generated by the minimum loss along the GIRP curve (GIRP Min Poisson) is compared with the final GIRP model (GIRP Final Poisson) and with the Poisson regression model.  In practice, one would only consider predictions using the regularized model, but here we want to compare against the unregularized model as well.  Table \ref{table:simulation_poisson_stats} demonstrates that Poisson isotonic regression works well with a reasonable number of variables (2-5 for the first simulation and 2-3 for the second simulation), however is outperformed by the simple Poisson regression with more than 5 variables.  In comparing the regularized GIRP model with the final GIRP model, there is no statistical difference in this example.  The next simulation clearly exemplifies the effect of regularization, in addition to the use of generalized isotonic regression.

The second two examples use $\delta$-Huber loss as the loss function for generating models.  Data for the two simulations is generated as $x_{ij}\sim\mathcal{U}[0,3]$ and $x_{ij}\sim\mathcal{U}[0,5]$, respectively.  The $i^{th}$ observation in each simulation is generated as $y_i=(\prod_j{x_{ij}})+\mathcal{N}(0,d^2)$ and $y_i=(\sum_j{x^2_{ij}})+\mathcal{N}(0,(1.5d)^2)$, respectively, where $d$ is the dimension.  For a randomly chosen $0.5\%$ of the training data, the observations are multiplied by a factor of 20.  The generalized isotonic models are compared to the results of a Huber regression, and performance here is measured by mean squared error.  Note that we assume that squared error loss represents the true objective performance; the models are fit using Huber loss in order to avoid sensitivity to outliers. The regularized model generated by the minimum loss along the GIRP curve (GIRP Min Huber) is compared with the final GIRP model (GIRP Final Huber) and with the Huber regression model.  Table \ref{table:simulation_huber_stats} demonstrates that Huber isotonic regression works well with a reasonable number of variables (2-5 for the first simulation and 2-4 for the second simulation), however, again, a simple Huber regression outperforms GIRP for higher dimensions due to overfitting.  An important note here is the effect of regularization.  The average loss using the unregularized isotonic model is not statistically superior at any dimension to the average loss using a Huber regression while the regularized isotonic model produces statistically improved performance.

Figures \ref{fig:simulations_poisson} and \ref{fig:simulations_huber} display regularization paths for the Poisson and Huber simulations, respectively.  Each curve shows the performance from using increasingly complex models generated by GIRP.  Take, for example, the first curve ($d=2$) under Model 1 in Figure \ref{fig:simulations_poisson}.  The x-axis states the number of partitions in the particular GIRP model and the y-axis measures the negative Poisson log-likelihood of using the GIRP models (trained on the subtraining data) to predict the validation data.  As the number of partitions increases (i.e., as the model becomes more complex), performance improves.  Consider next $d=5$ under the same model.  After 12 iterations of GIRP the performance begins to worsen (the minimum along each curve is shown by a diamond).  This is exactly the effect of regularization.  Performance improves as the model complexity increases up to a certain point at which increasing the complexity further overfits the model and performance declines.  Thus, as done to obtain the performance in Tables \ref{table:simulation_poisson_stats} and \ref{table:simulation_huber_stats}, the model along the path that gives the best performance on the validation data is used to make predictions on the independent testing data.

The curves in Figure \ref{fig:simulations_huber} show similar paths for the generalized isotonic regressions with Huber loss where  performance is measured by mean squared error.  Here the effects of regularization are much more pronounced than they are in the Poisson simulations.  This suggests that robust regressions on applications where isotonicity is desired would greatly benefit from the regularization of GIRP with Huber loss.  We next exhibit this robustness effect on a data set for predicting the miles-per-gallon of automobiles.

\begin{table}[h!]
\begin{center}
\small{
\begin{tabular}{|c|c|c|c|c|c|}
\multicolumn{6}{c}{Model 1: $y_i\sim\mbox{Poisson}(\prod_j{\sqrt{x_{ij}}})$ with $x_{ij}\sim\mathcal{U}[0,10]$} \\ \hline
   Dim& GIRP Min Poisson& GIRP Final Poisson&Poisson Regression &  Min & GIRP \\
   & Neg. Log-Likelihood & Neg. Log-Likelihood & Neg. Log-Likelihood & Path & Length \\
  \hline
2&\textbf{30678.65 ($\pm$ 23.83)}&\textbf{30678.73 ($\pm$ 23.84)}&32031.44 ($\pm$ 24.22)&217&298\\ \hline
3&\textbf{36395.90 ($\pm$ 38.03)}&\textbf{36397.99 ($\pm$ 37.84)}&40421.91 ($\pm$ 37.30)&90&618\\ \hline
4&43908.67 ($\pm$ 54.93)&43944.68 ($\pm$ 56.47)&54108.80 ($\pm$ 60.77)&78&584\\ \hline
5&\textbf{66812.67 ($\pm$ 240.06)}&\textbf{68096.06 ($\pm$ 347.41)}&81096.57 ($\pm$ 132.33)&30&371\\ \hline
6&200068.80 ($\pm$ 1072.14)&220308.20 ($\pm$ 2059.38)&140478.56 ($\pm$ 398.51)&9&479\\ \hline
 \multicolumn{6}{c}{} \\
\multicolumn{6}{c}{Model 2: $y_i\sim\mbox{Poisson}(\sum_j{x^2_{ij}})$ with $x_{ij}\sim\mathcal{U}[5,10]$} \\ \hline
   Dim& GIRP Min Poisson& GIRP Final Poisson&Poisson Regression &  Min & GIRP \\
   & Neg. Log-Likelihood & Neg. Log-Likelihood & Neg. Log-Likelihood & Path & Length \\
  \hline
2&\textbf{56957.75 ($\pm$ 31.19)}&\textbf{56957.77 ($\pm$ 31.20)}&57802.85 ($\pm$ 31.17)&661&794\\ \hline
3&\textbf{60650.13 ($\pm$ 30.83)}&\textbf{60650.62 ($\pm$ 31.32)}&60861.90 ($\pm$ 29.95)&105&1239\\ \hline
4&64008.55 ($\pm$ 40.17)&64041.56 ($\pm$ 44.10)&62956.84 ($\pm$ 23.68)&57&1105\\ \hline
5&67837.30 ($\pm$ 50.78)&68182.67 ($\pm$ 78.18)&64590.51 ($\pm$ 30.33)&16&806\\ \hline
6&74438.33 ($\pm$ 97.72)&75479.73 ($\pm$ 91.07)&65936.54 ($\pm$ 28.06)&16&544\\ \hline
 \end{tabular} }
\end{center}
\caption{Statistics for count data simulations generated by two different models as labeled above.  {\em GIRP Min (Final) Poisson Neg. Log-Likelihood} (LL)
refers to the negative Poisson LL of predicting on independent testing data using the model that produced the minimum (final) loss along a regularization path generated on training data.  {\em Min Path} is
the number of partitions made to generate the minimum negative Poisson LL and GIRP
Path Length is the number of partitions in the global generalized isotonic
solution. {\em Poisson Regression Neg. LL} is the negative Poisson LL from using Poisson
regressions. Bolded MSE values for minimum and final GIRP negative Poisson LL
indicate that they are significantly lower than the negative LL of the Poisson regression at level $.05$.} \label{table:simulation_poisson_stats}
\end{table}

\begin{table}[h!]
\begin{center}
\small{
\begin{tabular}{|c|c|c|c|c|c|}
\multicolumn{6}{c}{Model 1: $y_i=(\prod_j{x_{ij}})+\mathcal{N}(0,d^2)$ with $x_{ij}\sim\mathcal{U}[0,3]$} \\ \hline
   Dim& GIRP Min Huber& GIRP Final Huber&Huber Regression &  Min & GIRP \\
   & MSE & MSE & MSE & Path & Length \\
  \hline
2&4.21 ($\pm$ 0.30)&4.35 ($\pm$ 0.36)&4.55 ($\pm$ 0.03)&49&421\\ \hline
3&\textbf{9.69 ($\pm$ 0.07)}&\textbf11.71 ($\pm$ 2.44)&13.18 ($\pm$ 0.10)&27&1607\\ \hline
4&\textbf{22.93 ($\pm$ 0.26)}&90.83 ($\pm$ 64.78)&36.94 ($\pm$ 0.51)&10&3645\\ \hline
5&\textbf{83.20 ($\pm$ 1.23)}&280.08 ($\pm$ 106.02)&115.47 ($\pm$ 2.43)&6&5783\\ \hline
6&370.56 ($\pm$ 10.41)&2080.71 ($\pm$ 915.59)&391.01 ($\pm$ 12.09)&3&7531\\ \hline
 \multicolumn{6}{c}{} \\
\multicolumn{6}{c}{Model 2: $y_i=(\sum_j{x^2_{ij}})+\mathcal{N}(0,(1.5d)^2)$ with $x_{ij}\sim\mathcal{U}[0,5]$} \\ \hline
   Dim&Huber Min Huber& GIRP Final Huber&Huber Regression &  Min & GIRP \\
   & MSE & MSE & MSE & Path & Length \\
  \hline
2&\textbf{9.60 ($\pm$ 0.07)}&14.12 ($\pm$ 8.35)&15.98 ($\pm$ 0.11)&57&1154\\ \hline
3&\textbf{23.79 ($\pm$ 0.20)}&60.17 ($\pm$ 34.96)&30.61 ($\pm$ 0.23)&33&3283\\ \hline
4&\textbf{48.20 ($\pm$ 0.41)}&193.37 ($\pm$ 83.58)&50.02 ($\pm$ 0.36)&16&5705\\ \hline
5&85.62 ($\pm$ 0.58)&599.59 ($\pm$ 298.06)&73.44 ($\pm$ 0.54)&8&7785\\ \hline
6&145.06 ($\pm$ 1.34)&1602.43 ($\pm$ 620.45)&101.12 ($\pm$ 0.73)&8&9283\\ \hline
  \end{tabular} }
\end{center}
\caption{Statistics for count data simulations generated by two different models as labeled above.  {\em GIRP Min (Final) Huber MSE} refers to the MSE of predicting on independent testing data using the model that produced the minimum (final) loss along a regularization path generated on training data.  {\em Min Path} is
the number of partitions made to generate the minimum MSE and GIRP
Path Length is the number of partitions in the global generalized isotonic
solution. {\em Huber Regression MSE} is the MSE from using Huber
regressions.} \label{table:simulation_huber_stats}
\end{table}

\begin{figure} [h!]
\begin{tabular*}{\textwidth}{@{\extracolsep{\fill}}cc}
\textbf{Model 1}& \textbf{Model 2}\\
\psfrag{m}[b][t]{\scriptsize{N-Poiss Loss}}
\psfrag{d}[b]{\scriptsize{$d=2$}}
\includegraphics[width=.5\textwidth]{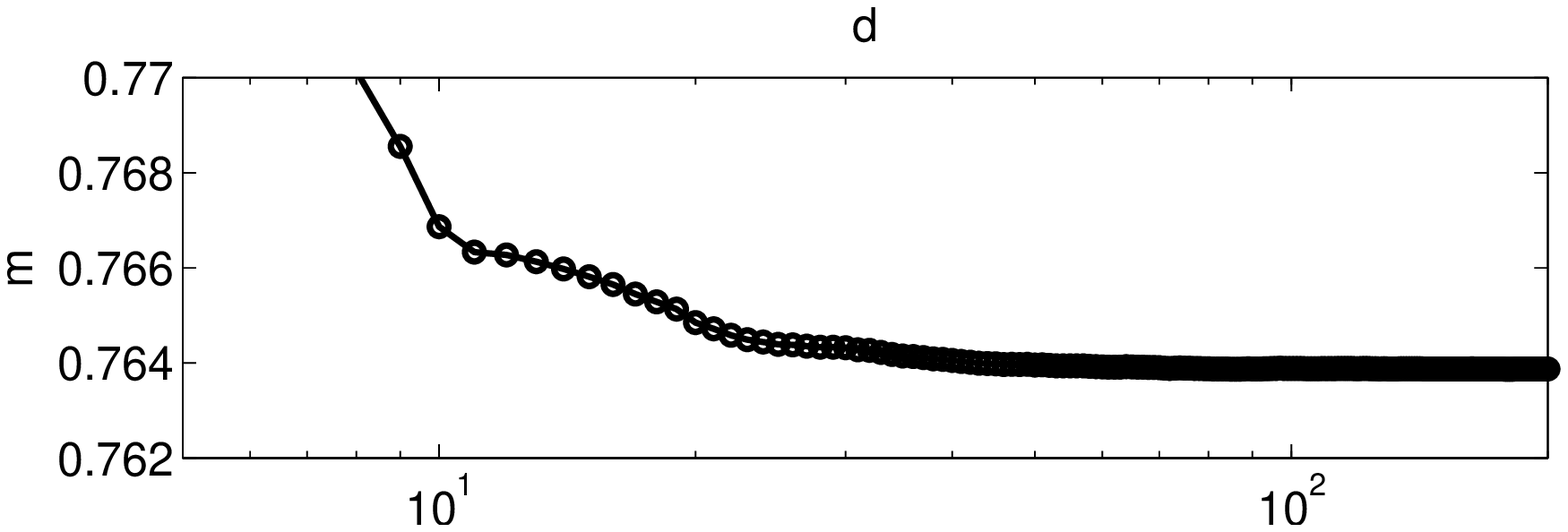}
& \psfrag{m}[b][t]{\scriptsize{N-Poiss Loss}}
\psfrag{d}[b]{\scriptsize{$d=2$}}
\includegraphics[width=.5\textwidth]{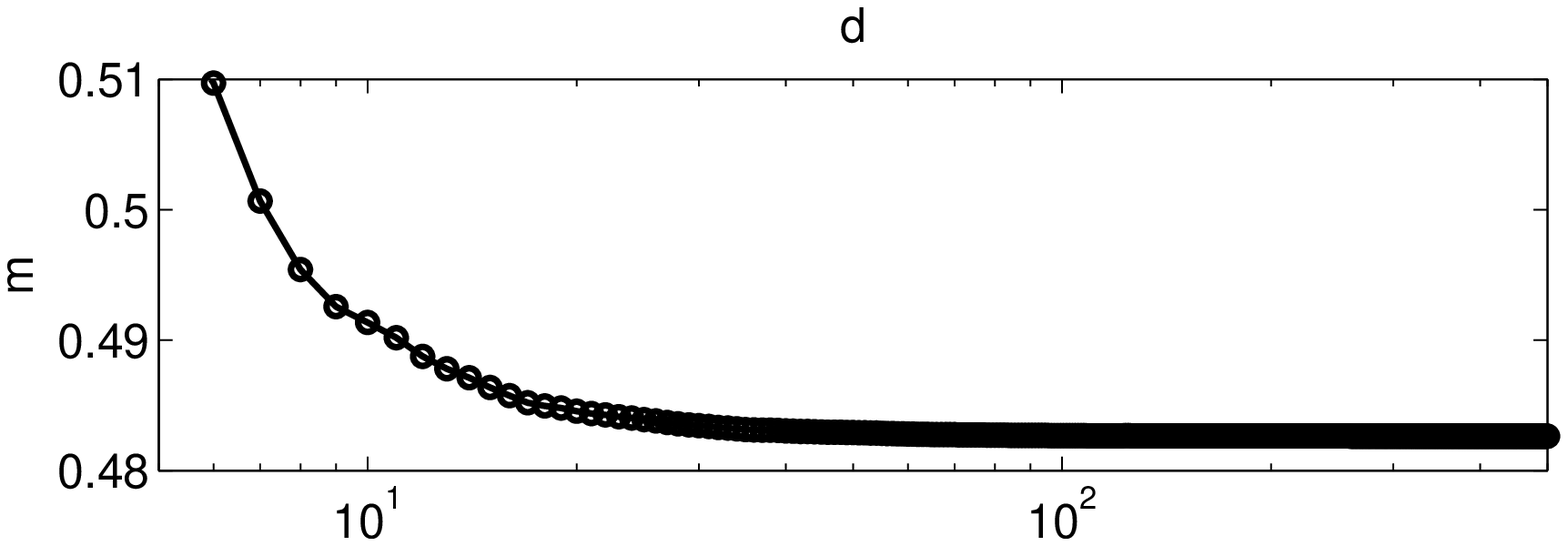}
\\\\
\psfrag{m}[b][t]{\scriptsize{N-Poiss Loss}}
\psfrag{d}[b]{\scriptsize{$d=3$}}
\includegraphics[width=.5\textwidth]{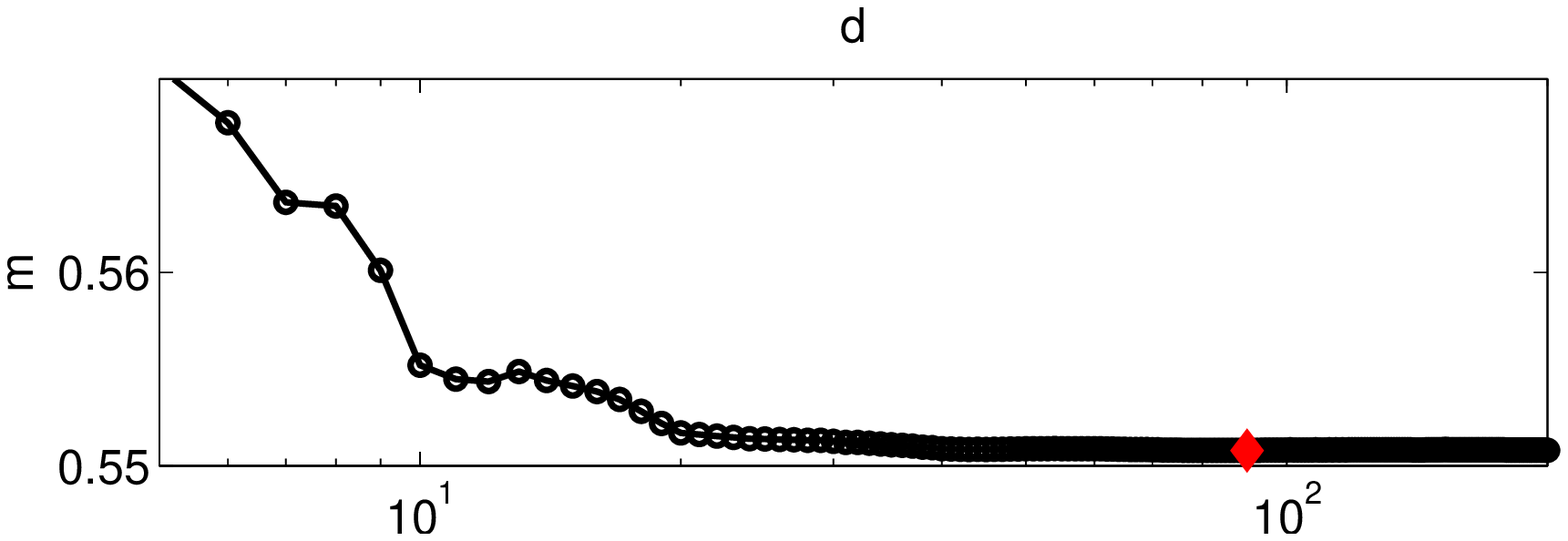}
& \psfrag{m}[b][t]{\scriptsize{N-Poiss Loss}}
\psfrag{d}[b]{\scriptsize{$d=3$}}
\includegraphics[width=.5\textwidth]{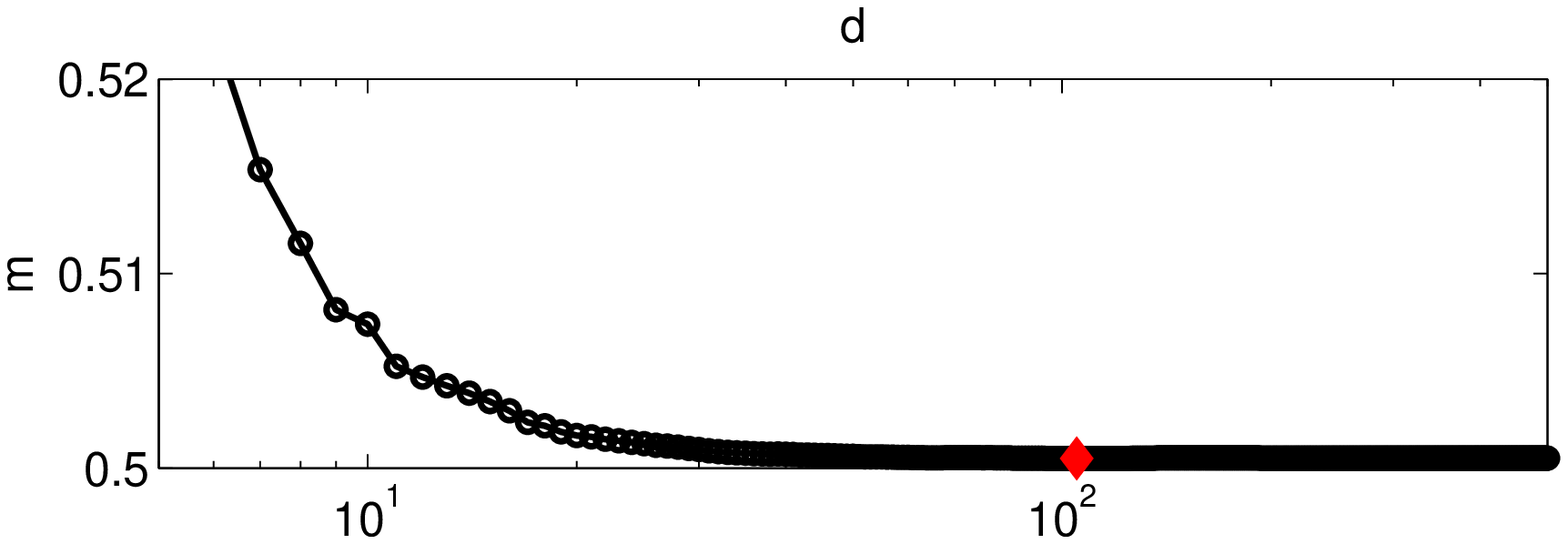}
\\\\
\psfrag{m}[b][t]{\scriptsize{N-Poiss Loss}}
\psfrag{d}[b]{\scriptsize{$d=4$}}
\includegraphics[width=.5\textwidth]{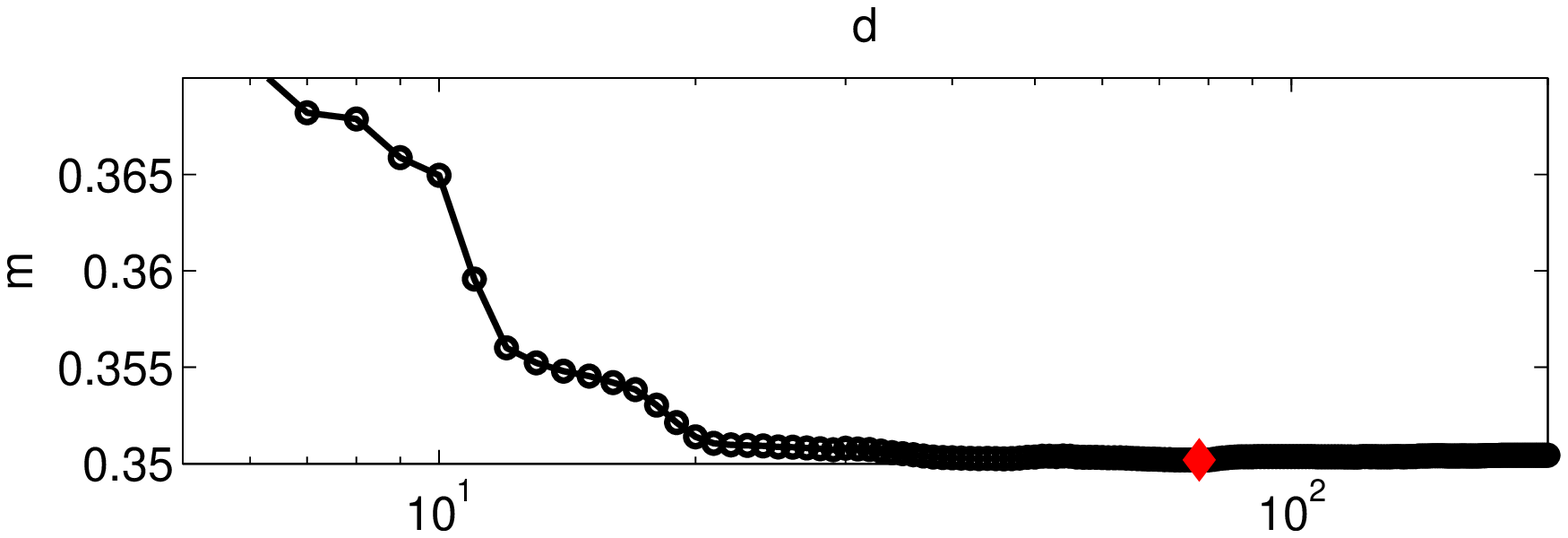}
& \psfrag{m}[b][t]{\scriptsize{N-Poiss Loss}}
\psfrag{d}[b]{\scriptsize{$d=4$}}
\includegraphics[width=.5\textwidth]{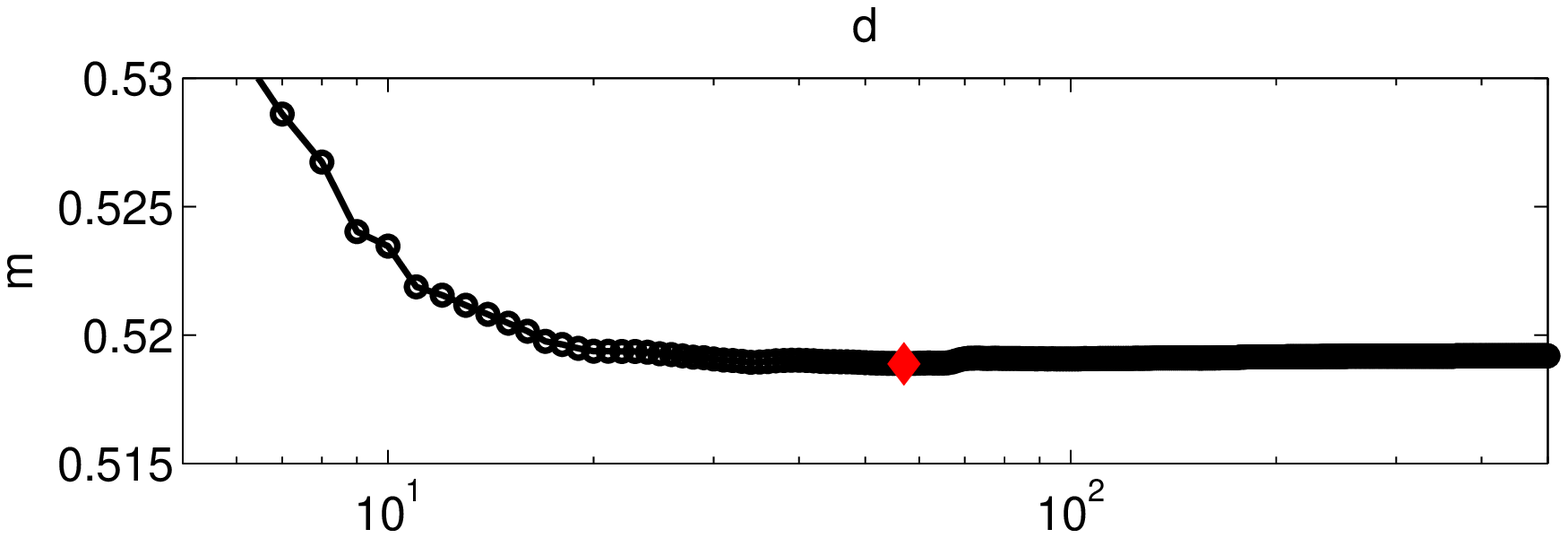}
\\\\
\psfrag{m}[b][t]{\scriptsize{N-Poiss Loss}}
\psfrag{d}[b]{\scriptsize{$d=5$}}
\includegraphics[width=.5\textwidth]{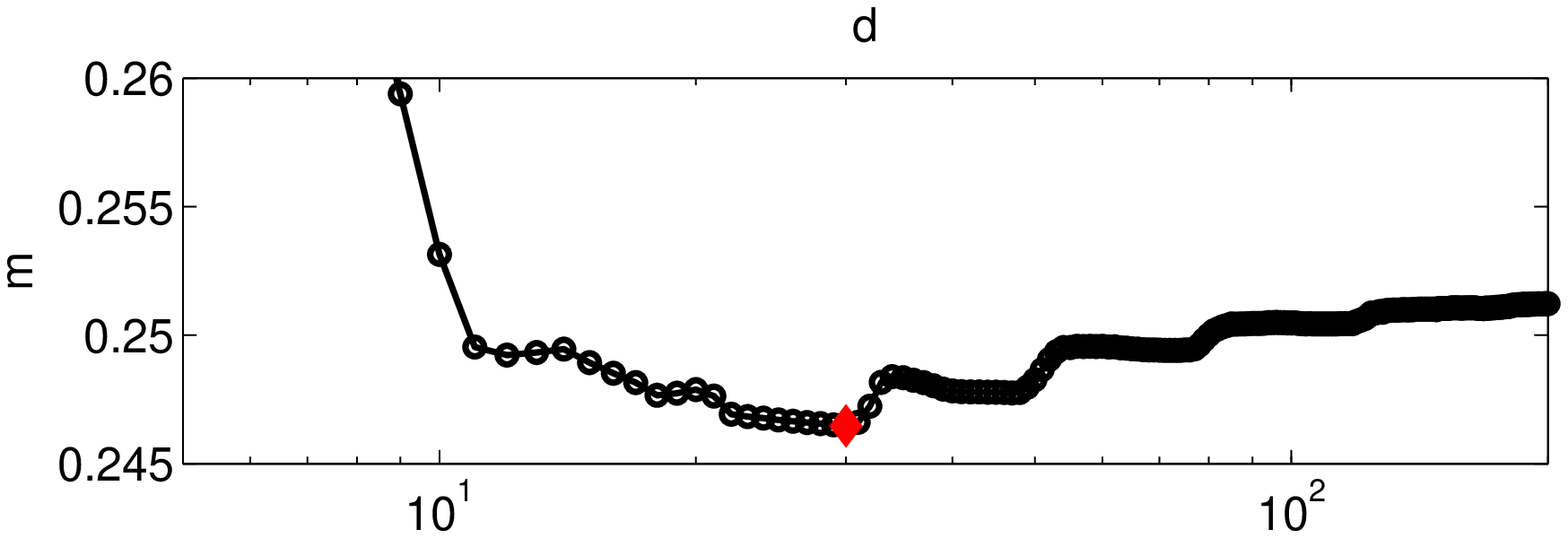}
& \psfrag{m}[b][t]{\scriptsize{N-Poiss Loss}}
\psfrag{d}[b]{\scriptsize{$d=5$}}
\includegraphics[width=.5\textwidth]{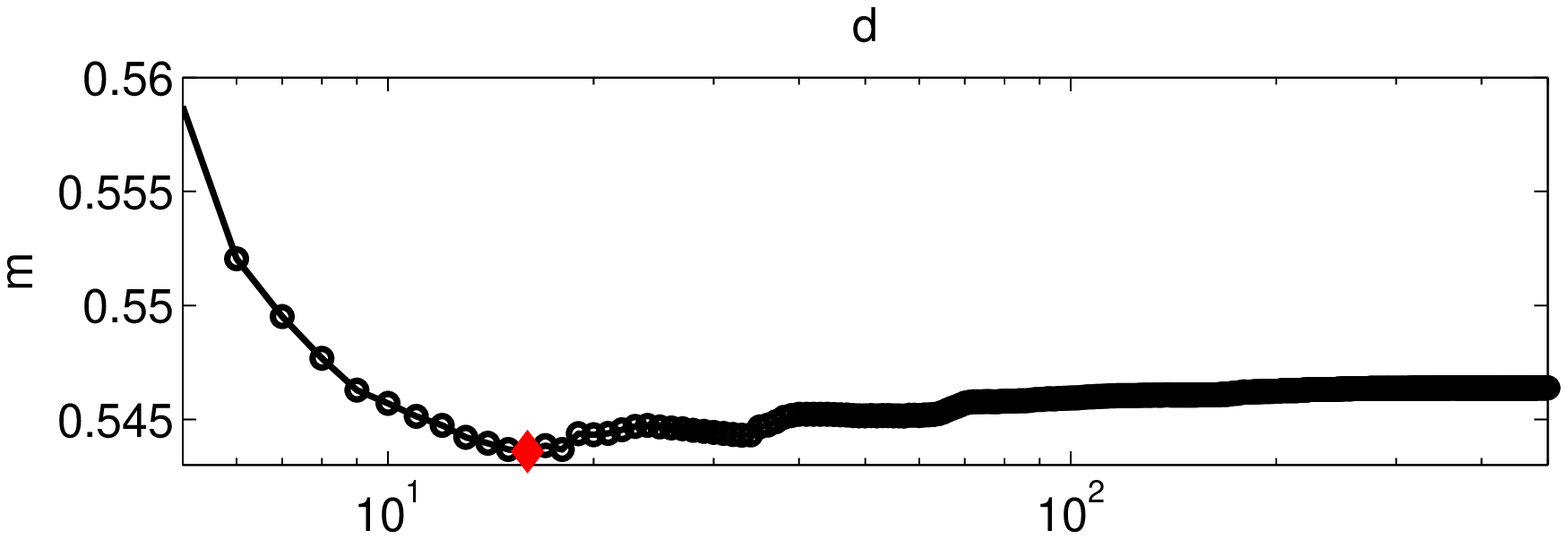}
\\\\
\psfrag{m}[b][t]{\scriptsize{N-Poiss Loss}}
\psfrag{d}[b]{\scriptsize{$d=6$}}
\includegraphics[width=.5\textwidth]{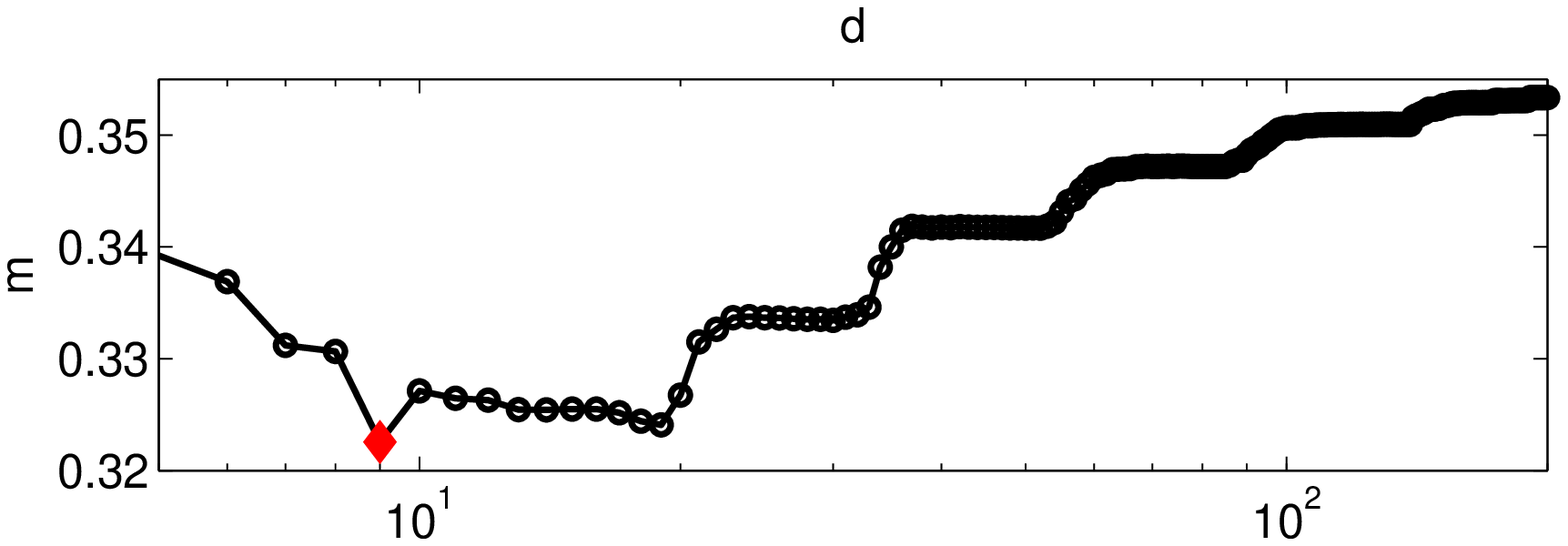}
& \psfrag{m}[b][t]{\scriptsize{N-Poiss Loss}}
\psfrag{d}[b]{\scriptsize{$d=6$}}
\includegraphics[width=.5\textwidth]{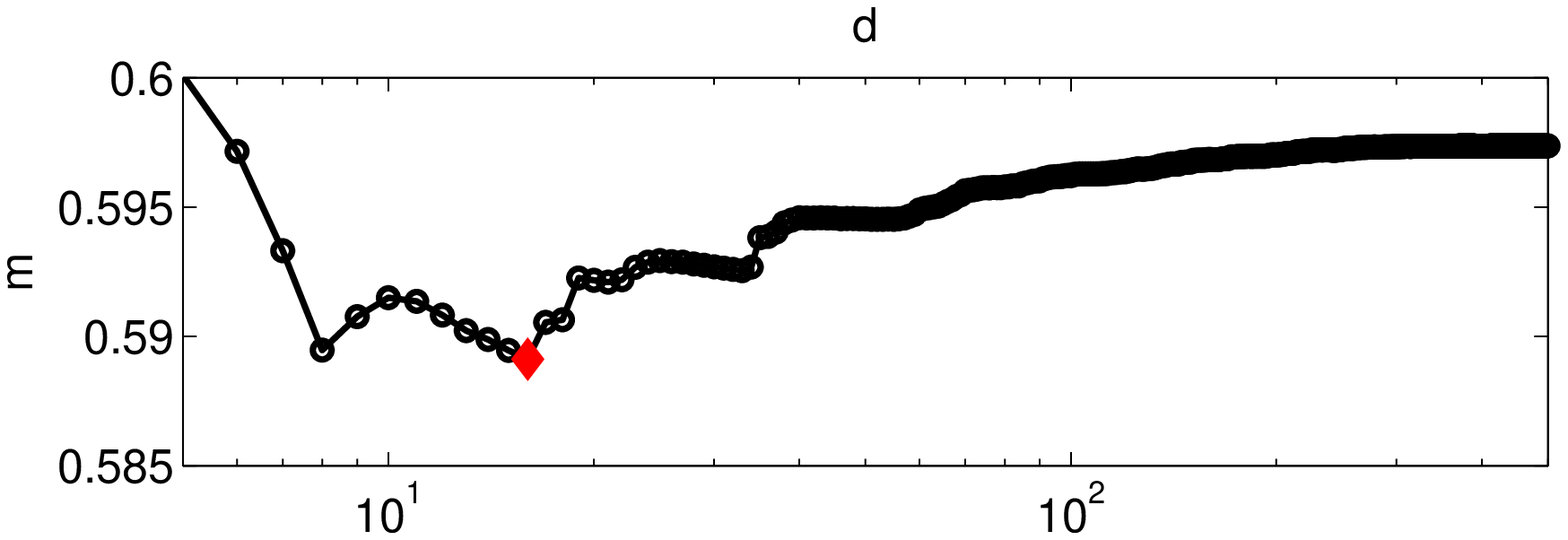}
 \end{tabular*}
\caption{Normalized negative Poisson log-likelihood (LL) for predictions of count data simulations with different dimensions $d$.  Each path is normalized by the loss of the initial model.  The x-axis in each figure corresponds to the number of partitions made by GIRP, i.e., the curves show how the
normalized negative Poisson LL of test data varies as the GIRP algorithm progresses. Model 1 uses the function $y_i\sim\mbox{Poisson}(\prod_j{\sqrt{x_{ij}}})$ with $x_{ij}\sim\mathcal{U}[0,10]$ and Model 2 uses the function $y_i\sim\mbox{Poisson}(\sum_j{x^2_{ij}})$ with $x_{ij}\sim\mathcal{U}[0,5]$.  Fifty simulations were run with 12000 training and 3000 testing points.  Only the first few hundred partitions of the paths are displayed in order to make the loss of the earlier GIRP
iterations visually clearer.  Scales also differ in order to make
the shapes of the curves clear.} \label{fig:simulations_poisson}
\end{figure}

\begin{figure} [h!]
\begin{tabular*}{\textwidth}{@{\extracolsep{\fill}}cc}
\textbf{Model 1}& \textbf{Model 2}\\
\psfrag{m}[b][t]{\scriptsize{N-MSE}}
\psfrag{d}[b]{\scriptsize{$d=2$}}
\includegraphics[width=.5\textwidth]{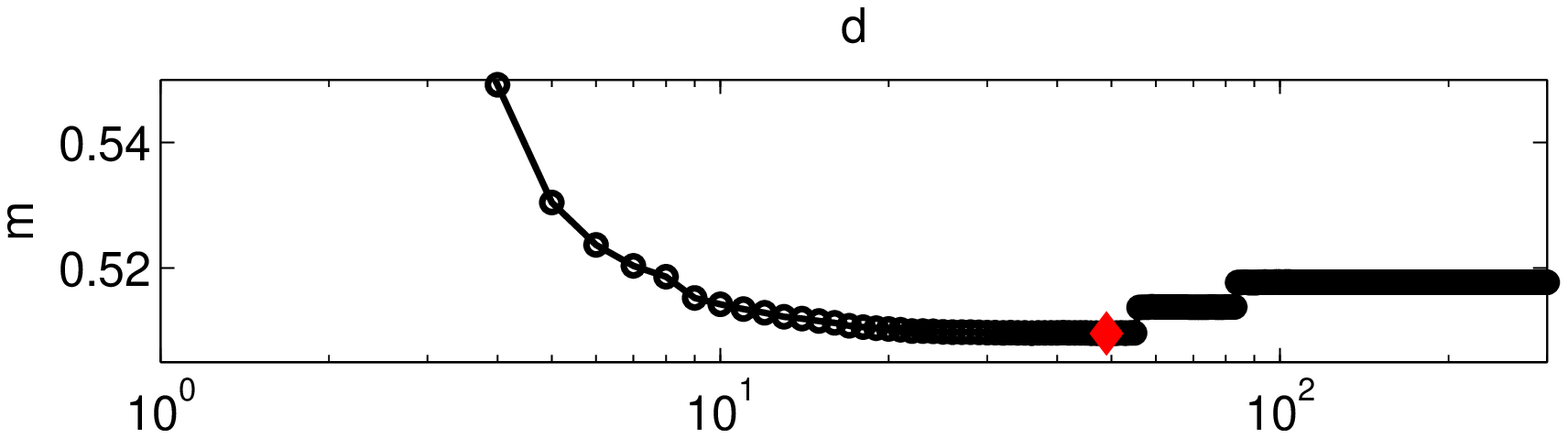}
& \psfrag{m}[b][t]{\scriptsize{N-MSE}}
\psfrag{d}[b]{\scriptsize{$d=2$}}
\includegraphics[width=.5\textwidth]{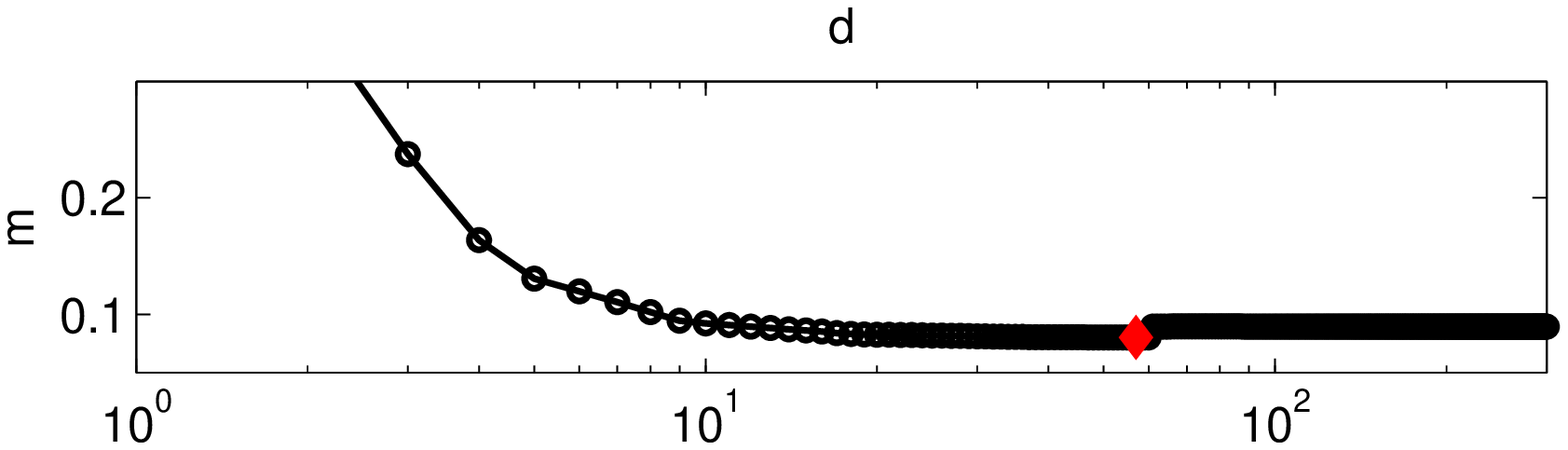}
\\\\
\psfrag{m}[b][t]{\scriptsize{N-MSE}}
\psfrag{d}[b]{\scriptsize{$d=3$}}
\includegraphics[width=.5\textwidth]{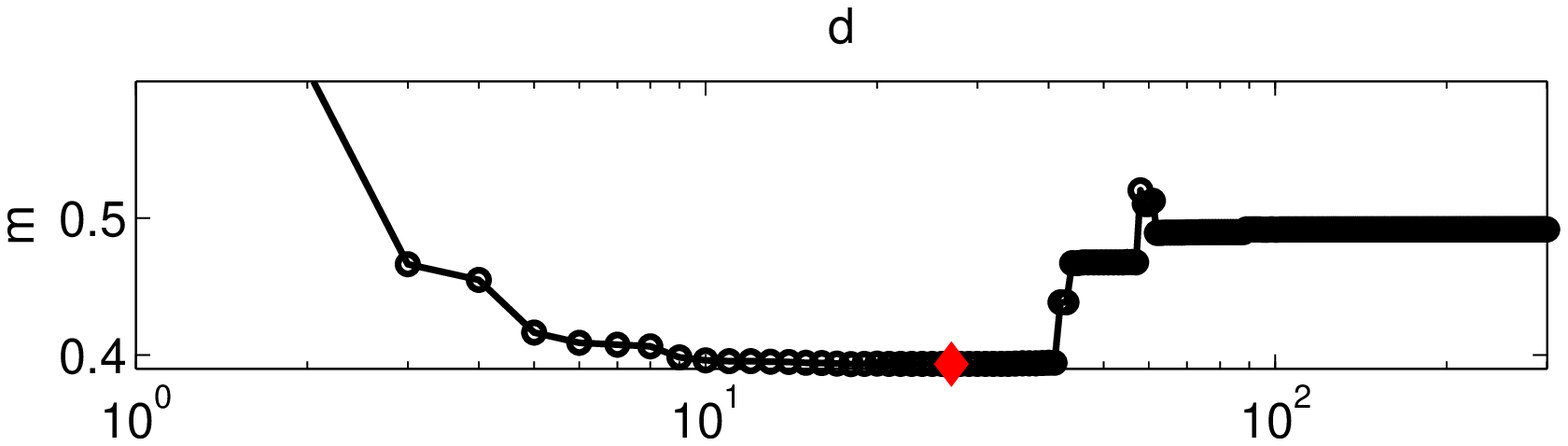}
& \psfrag{m}[b][t]{\scriptsize{N-MSE}}
\psfrag{d}[b]{\scriptsize{$d=3$}}
\includegraphics[width=.5\textwidth]{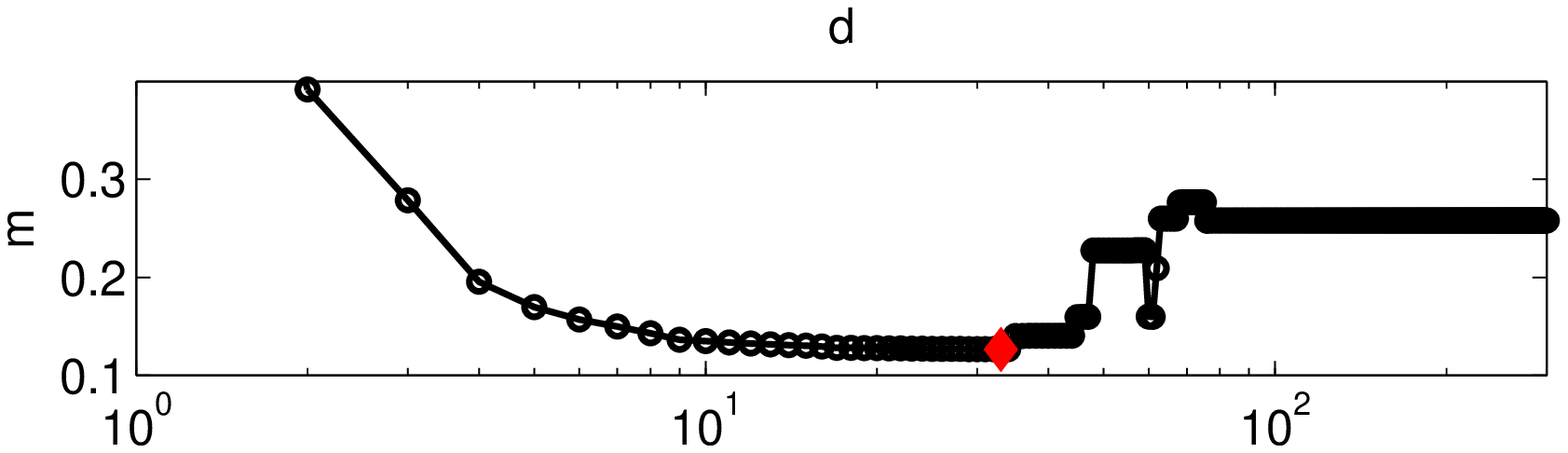}
\\\\
\psfrag{m}[b][t]{\scriptsize{N-MSE}}
\psfrag{d}[b]{\scriptsize{$d=4$}}
\includegraphics[width=.5\textwidth]{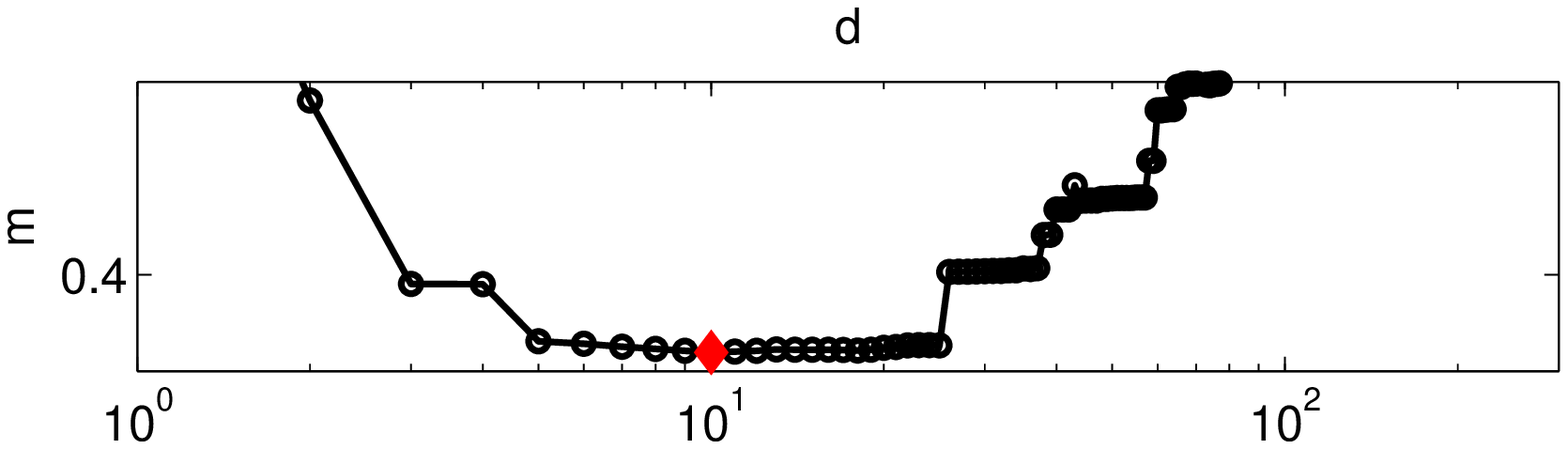}
& \psfrag{m}[b][t]{\scriptsize{N-MSE}}
\psfrag{d}[b]{\scriptsize{$d=4$}}
\includegraphics[width=.5\textwidth]{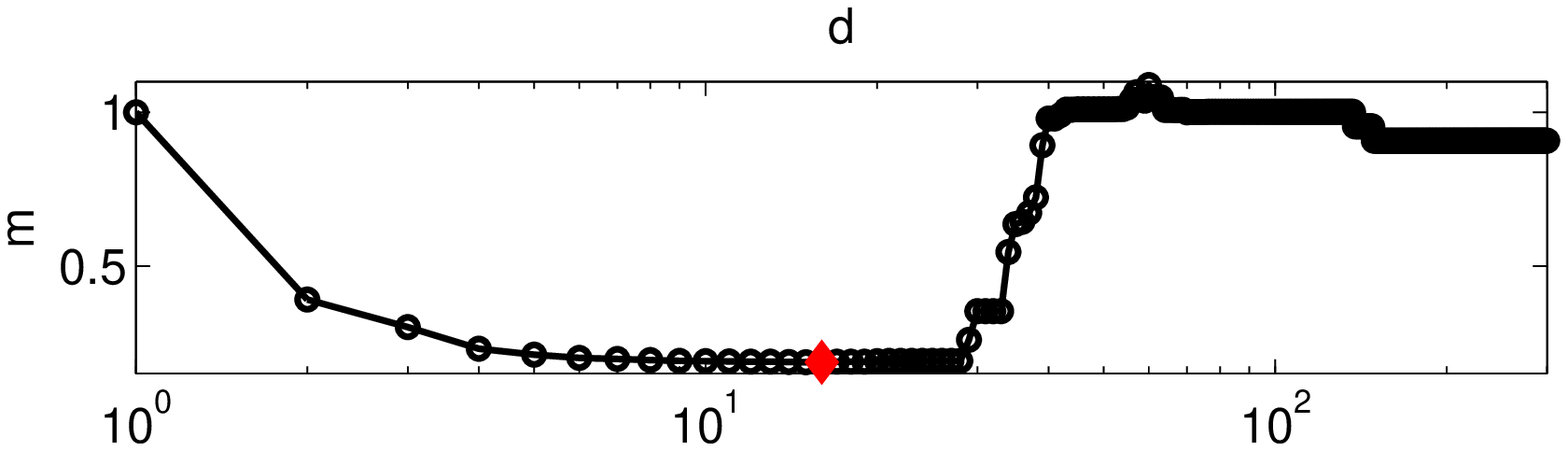}
\\\\
\psfrag{m}[b][t]{\scriptsize{N-MSE}}
\psfrag{d}[b]{\scriptsize{$d=5$}}
\includegraphics[width=.5\textwidth]{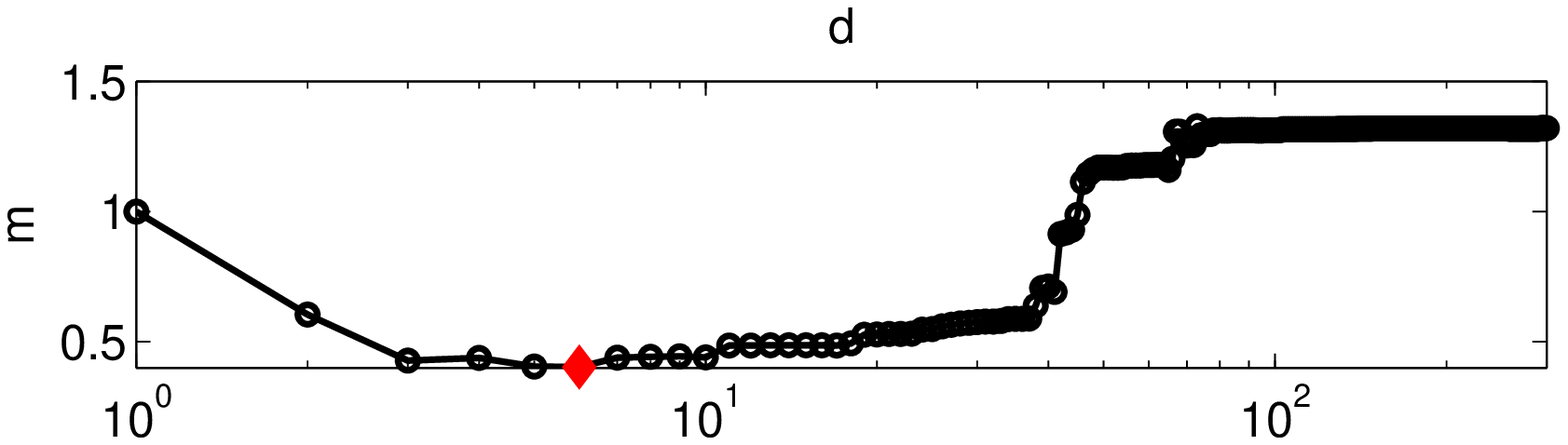}
& \psfrag{m}[b][t]{\scriptsize{N-MSE}}
\psfrag{d}[b]{\scriptsize{$d=5$}}
\includegraphics[width=.5\textwidth]{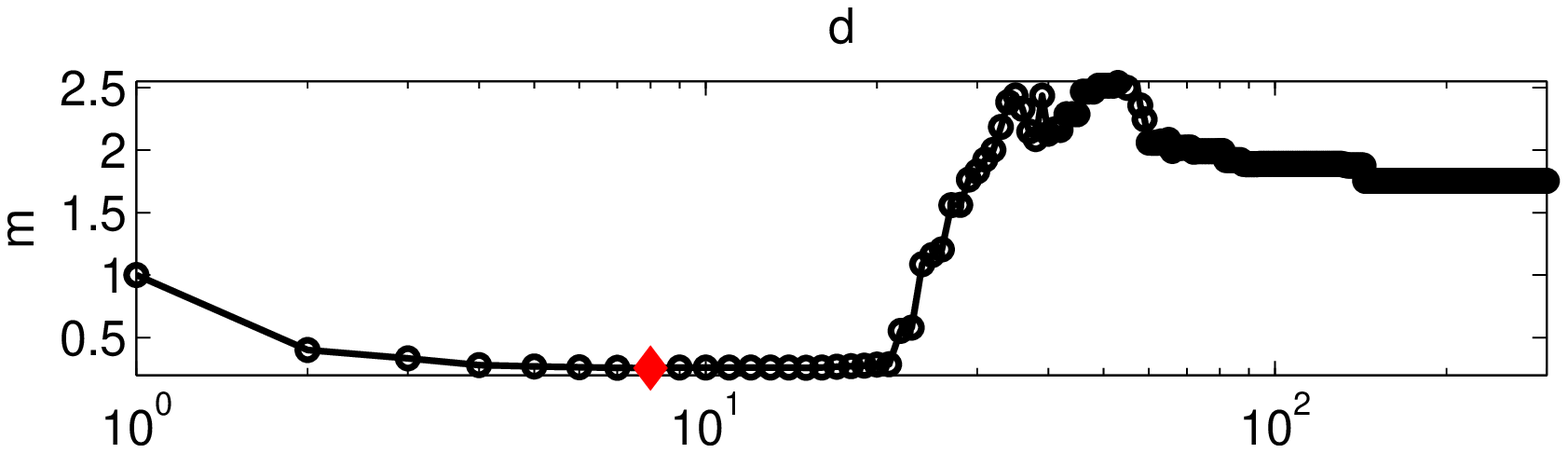}
\\\\
\psfrag{m}[b][t]{\scriptsize{N-MSE}}
\psfrag{d}[b]{\scriptsize{$d=6$}}
\includegraphics[width=.5\textwidth]{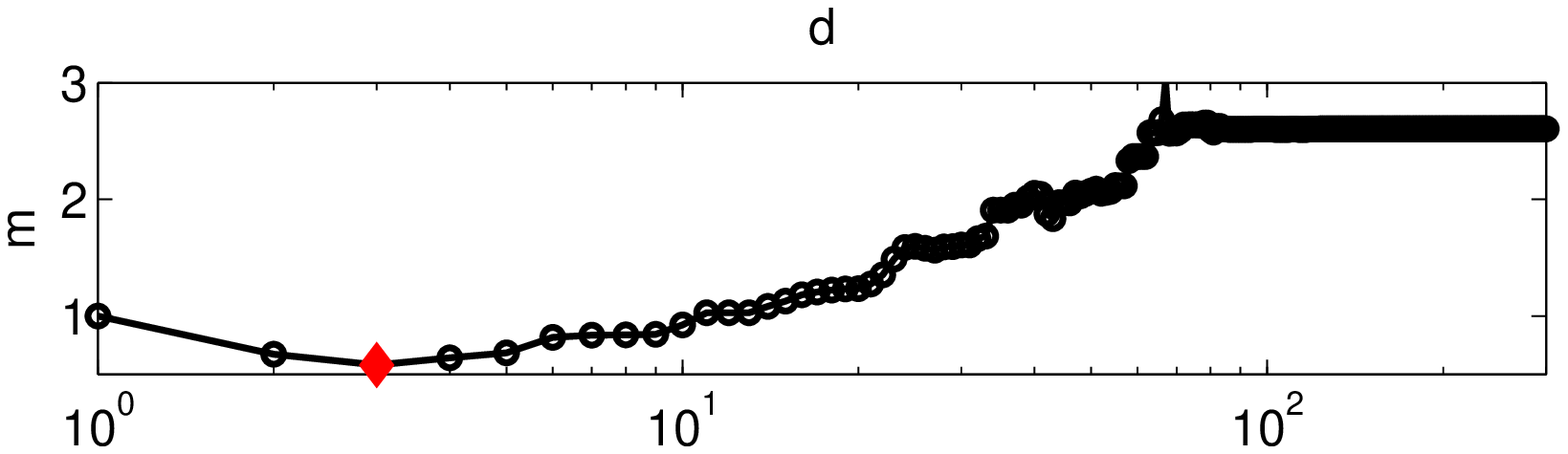}
& \psfrag{m}[b][t]{\scriptsize{N-MSE}}
\psfrag{d}[b]{\scriptsize{$d=6$}}
\includegraphics[width=.5\textwidth]{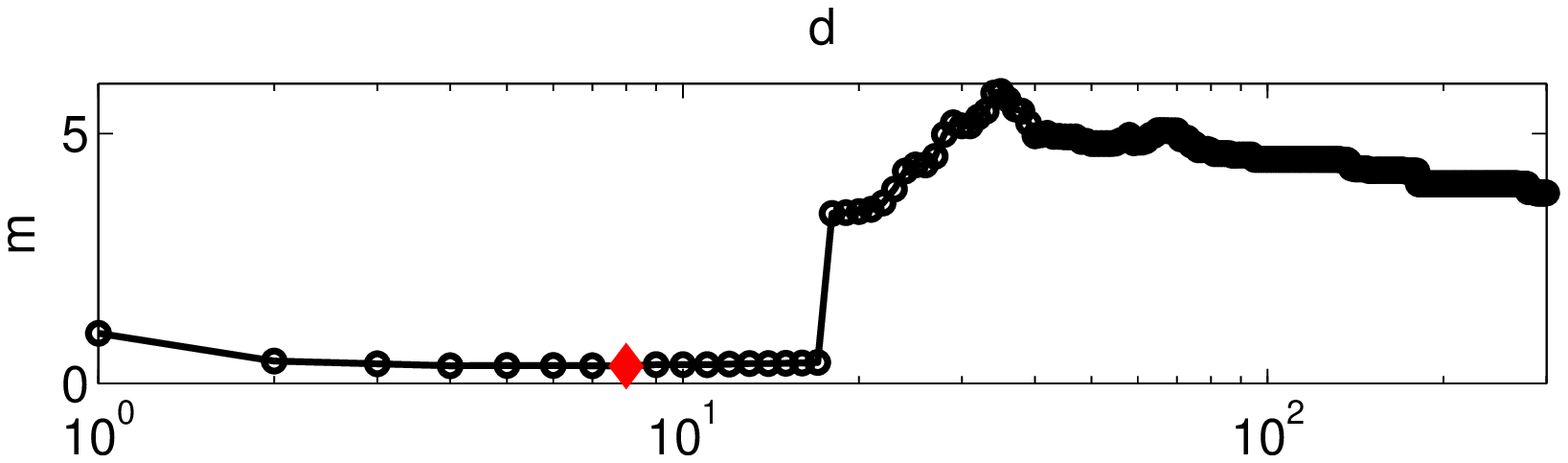}
 \end{tabular*}
\caption{Normalized MSE for predictions
of two isotonic simulations with different dimensions $d$.  Each path is normalized by the MSE of the initial model.  The x-axis in each figure corresponds to the
number of partitions made by GIRP, i.e., the curves show how the
normalized MSE of test data varies as the GIRP algorithm progresses.
Model 1 uses the function $y_i=(\prod_j{x_{ij}})+\mathcal{N}(0,d^2)$ with $x_{ij}\sim\mathcal{U}[0,10]$ and Model 2 uses the function $y_i=(\sum_j{x^2_{ij}})+\mathcal{N}(0,(1.5d)^2)$ with $x_{ij}\sim\mathcal{U}[0,5]$.  Fifty simulations were run with 12000 training and 3000 testing points.  Only the first few.} \label{fig:simulations_huber}
\end{figure}

\subsection{Predicting Miles-Per-Gallon}\label{ss:perf_simulations}
The next example compares the $l_2$ versus Huber loss regressions.
Note that Huber loss function is an example that cannot be solved
using the theory of \citeasnoun{Barl1972} and $l_2$ isotonic
regression.  This example uses a data set of 392 automobiles
\cite{Asun2010} and models miles-per-gallon using the following
seven variables: origin, model year, number of cylinders,
acceleration, displacement, horsepower, and weight.  Isotonic
regression with an $l_2$ loss function was already shown to be
useful for this data set in \citeasnoun{Luss2012}.  In this
experiment, we have modified one data point to be an outlier.  This
random point is chosen such that isotonicity constraints require its
fit to be less than the fits of five other data points.  The experiment simulates a real-life outlier problem which affects the training sample but should not affect prediction.  We assume squared error loss to be the true prediction criterion (therefore
the out-of-sample evaluation criterion), and fit models with Huber
loss to avoid sensitivity to outliers. Table \ref{table:auto_stats}
displays the results of one random division of the data (2/3 for
training, 1/3 for testing). A paired t-test comparing the
out-of-sample predictive performance  of the two models (IRP and
GIRP) confirms the significant edge of the model generated with
Huber's loss function in this setting.

\begin{table}[h!]
\begin{center}
\small{
\begin{tabular}{|c|c|c|c|c|c|c|}
\hline
  Number & IRP LS & GIRP Huber & IRP LS & IRP LS & GIRP Huber & GIRP Huber\\
  Variables & Min MSE &Min MSE &Min Path &Path Length & Min Path &Path Length\\
  \hline
1&37.47 $\pm$ 9.67&38.48 $\pm$ 10.42&2&3&2&3\\ \hline 
2&31.22 $\pm$ 7.04&\textbf{27.01 $\pm$ 6.28}&17&17&7&16\\ \hline 
3&21.19 $\pm$ 6.75&15.83 $\pm$ 4.76&12&30&5&30\\ \hline 
4&22.90 $\pm$ 6.78&\textbf{15.53 $\pm$ 4.01}&4&53&11&54\\ \hline 
5&19.94 $\pm$ 7.06&\textbf{10.95 $\pm$ 3.15}&4&78&29&81\\ \hline 
6&17.55 $\pm$ 5.99&\textbf{9.78 $\pm$ 2.89}&4&86&69&90\\ \hline 
7&18.91 $\pm$ 6.29&\textbf{10.24 $\pm$ 3.51}&4&95&71&95\\ \hline 
\end{tabular}}
\end{center}
\caption{\label{table:auto_stats} Statistics for auto mpg data.
Miles-per-gallon is regressed on a seven potential variables:
origin, model year, number of cylinders, acceleration, displacement,
horsepower, and weight. A comparison between the results of IRP and
GIRP with Huber loss is shown. One data point $y_i$ such that
$x_i\preceq x_j$ for $j=1\ldots 5$ was modified to be an outlier.
$\delta$ for Huber loss is set to one standard deviation of the
training responses.  Bold demonstrates statistical significance with
95\% confidence determined by a paired t-test using 131
out-of-sample observed squared losses obtained from models trained
on 261 in-sample observations.}
\end{table}

\section{Conclusion}
In this paper, we show how a relatively minor adjustment to the
previously proposed IRP algorithm leads to a generalization allowing us to efficiently fit isotonic models under any convex
differentiable loss function. Our proposed GIRP algorithm also
generates regularized isotonic solutions along its path, in addition
to the optimal isotonic solution. An important remaining challenge
is to generalize the approach to handling convex non-differentiable
loss functions (like absolute loss or the hinge loss of support
vector machines), an important topic for future
research.  Our analysis does not hold in this case due to nonuniqueness of the subproblems.

{\small
\bibliographystyle{agsm}
\bibliography{ConvexIsotonic}}

\section{Appendix}
We need the following additional terminology: A group $X$ \emph{majorizes} (\emph{minorizes}) another group $Y$ if
$X\succeq Y$ ($X\preceq Y$).  A group $X$ is a \emph{majorant}
(\emph{minorant}) of $X\cup A$ where $A=\cup_{i=1}^k{A_i}$ if
$X\not\prec A_i$ ($X\not\succ A_i$) $\forall i=1\ldots k$.

\noindent \\
\textbf{Theorem \ref{th:no_regret_cut}:}\\
\begin{proof}
We prove by contradiction.  Assume there exists a union of $K$
blocks in $V$ in the optimal solution labeled
$\mathcal{M}=M_1\cup\ldots\cup M_K$ that get broken by the cut, with
$M_1$ and $M_K$ as the minorant and majorant block in $\mathcal{M}$,
and $M_{k}^L$ and $M_{k}^U$ as the groups in $M_k$ below and above
the cut.  Define $\mathcal{L}$ as the union of all blocks in $V$
that lie ``below'' the algorithm cut, $\mathcal{U}$ as the union of
all blocks in $V$ that lie ``above'' the algorithm cut. Further
define $A_K^L\subseteq\mathcal{L}$ ($A_1^U\subseteq\mathcal{U}$) as
the union of blocks along the algorithm cut such that $A_K^L\succ
M_K^L$ ($A_1^U\prec M_1^U$).  Figure \ref{fig:proof1} depicts an
example of these definitions where $A_1^U=A_1^L=A_K^U=A_K^L=\{\}$
for simplicity.

We first prove that $w_{M_1}>w_{V}$. First, consider the case
$A_1^U=\{\}$.  By convexity of $f_i(\cdot)$ and summing over group
$M_1^U$, we have
\[
\displaystyle\sum_{i\in M_1^U}{f_i(w_{M_1^U})}\ge
\displaystyle\sum_{i\in
M_1^U}{f_i(w_V)}+(w_{M_1^U}-w_V)\displaystyle\sum_{i\in
M_1^U}{\frac{\partial
f_i(\hat{y}_i)}{\partial\hat{y}_i}\bigg|_{w_V}}.
\]
Definition of the weight operator gives
\[
\displaystyle\sum_{i\in
M_1^U}{f_i(w_{M_1^U})}\leq\displaystyle\sum_{i\in
M_1^U}{f_i(w_V)}\Rightarrow(w_{M_1^U}-w_V)\displaystyle\sum_{i\in
M_1^U}{\frac{\partial
f_i(\hat{y}_i)}{\partial\hat{y}_i}\bigg|_{w_V}}\leq 0.
\]
Finally, by the definition of the algorithm cut in
(\ref{eq:optimal_cut_lp}) since no block exists below ${M_{1}^U}$ to
affect isotonicity,
\begin{equation} \label{eq:proof_thm1}
\displaystyle\sum_{i\in M_1^U}{\frac{\partial
f_i(\hat{y}_i)}{\partial\hat{y}_i}\bigg|_{w_V}}\leq 0
\end{equation}
so that $w_{M_1^U}\ge w_V$.  Since $M_1$ is a block, we have
$w_{M_{1}^L}> w_{M_{1}^U}$, and then \[ w_{M_1^L}>w_{M_1^U}>
w_{V}\Rightarrow w_{M_1}>w_{V}. \] For the case, $A_1^U\neq\{\}$, we
have $w_{M_1}>w_{A_1^U}>w_{V}$ with the first inequality due to
optimality and the second follows directly the proof above replacing
$M_1^U$ by $A_{1}^U$.  A proof for $w_{M_K} < w_{V}$ follows a
similar argument focusing on $M_K^L$.  Putting this together gives
$w_{M_1}>w_{V}>w_{M_K}$, which contradicts that $M_1$ and $M_K$ are
blocks in the global solution, since by assumption then
$w_{M_1}<w_{M_K}$.  The case $K=1$ is also trivially covered by the
above arguments.  We conclude that the algorithm cannot cut any
block.
\end{proof}

\begin{figure}[h!] \begin{center}
  \begin{tabular} {c}
     \includegraphics[width=0.49 \textwidth]{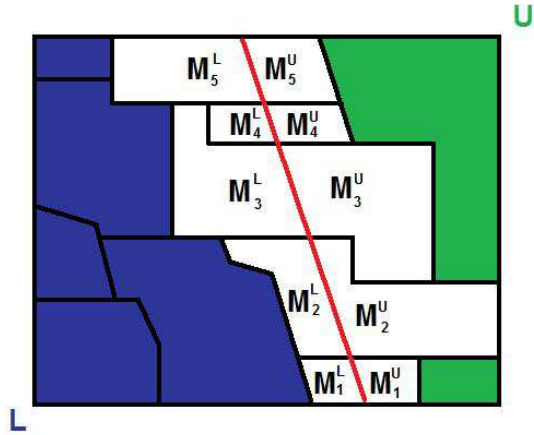}
  \end{tabular}
\caption{Illustration of proof of Theorem \ref{th:no_regret_cut}.
Black lines separate blocks.  The diagonal red line through the
center demonstrates a cut of Algorithm \ref{alg:GIRP}. $\mathcal{L}$
is the union of blue blocks below the cut and $\mathcal{U}$ is the
union of green blocks above the cut.  White blocks are blocks that
are potentially split by Algorithm \ref{alg:GIRP}.  These blocks are
split into $M_1^L,\ldots,M_5^L$ below the cut and
$M_1^U,\ldots,M_5^U$ above the cut.  In the proof, $M_i=M_i^L\cup
M_i^U$ $\forall i=1\ldots5$.  The proof shows, for example, that if
the algorithm splits $M_1$ into $M_1^L$ and $M_1^U$ according to the
defined cut in \ref{eq:optimal_cut_lp}, then there must be no
isotonicity violation when creating blocks from $M_1^L$ and $M_1^U$.
However, since $M_1$ is assumed to be a block, there must exist an
isotonicity violation between $M_1^L$ and $M_1^U$, providing a
contradiction. } \label{fig:proof1}
\end{center} \end{figure}

\noindent \\
\textbf{Theorem \ref{th:isotonic_solutions}:}\\
\begin{proof}
The proof is by induction.  The base case, i.e., first iteration,
where all points form one group is trivial.  The first cut is made
by solving linear program (\ref{eq:optimal_cut_lp}) which constrains
the solution to maintain isotonicity.

Assuming that iteration $k$ (and all previous iterations) provides
an isotonic solution, we prove that iteration $k+1$ must also
maintain isotonicity. Figure \ref{fig:proof2} helps illustrate the
situation described here.  Let $G$ be the group split at iteration
$k+1$ and denote $A$ ($B$) as the group under (over) the cut.  Let
$\mathcal{A}=\{X:X $ is a group at iteration $k+1, \exists i\in X$
 such that $(i,j)\in\mathcal{I}$ for some $j\in A\}$ (i.e.,
$X\in\mathcal{A}$ border $A$ from below).

Consider iteration $k+1$.  Denote $\mathcal{X}=\{X\in\mathcal{A}:w_A
< w_{X}\}$ (i.e., $X\in\mathcal{X}$ violates isotonicity with $A)$.
The split in $G$ causes the fit in nodes in $A$ to decrease.  Proof
that
\[
\displaystyle\sum_{i\in A}{\frac{\partial
f_i(\hat{y}_i)}{\partial\hat{y}_i} \bigg|_{w_G}}\geq 0
\]
follows the proof of (\ref{eq:proof_thm1}) in Theorem 1 above so
that $w_A\ge w_G$.  We will prove that when the fits in $A$
decrease, there can be no groups below $A$ that become violated by
the new fits to $A$, i.e.,  the decreased fits in $A$ cannot be such
that $\mathcal{X}\neq\{\}$.

We first prove that $\mathcal{X}=\{\}$ by contradiction. Assume
$\mathcal{X}\neq\{\}$. Denote $k_0<k+1$ as the iteration at which
the last of the groups in $\mathcal{X}$, denoted $D$, was split from
$G$ and suppose at iteration $k_0$, $G$ was part of a larger group
$H$ and $D$ was part of a larger group $F$.  It is important to note
that $X\bigcap(F\bigcup H)=\{\}$ $\forall X\in\mathcal{X}\setminus
D$ at iteration $k_0$ because by assumption all groups in
$\mathcal{X}\setminus D$ were separated from $A$ before iteration
$i$. Thus, at iteration $k_0$, $D$ is the only group bordering $A$
that violates isotonicity.

Let $D_U$ denote the union of $D$ and all groups in $F$ that
majorize $D$. By construction, $D_U$ is a majorant in $F$. Hence
$w_{D_U} < w_{F\cup H}$ by Algorithm \ref{alg:GIRP} and $w_A <
w_{D_U}$ by definition since $w_{D_U} > w_D > w_A$. Also by
construction, any set $X\in H$ that minorizes $A$ has $w_{X} <
w_{A}$ (each set $X$ that minorizes $A$ besides $D$ such that $w_{X}
< w_{A}$ has already been split from $A$). Hence we can denote $A_L$
as the union of $A$ and all groups in $H$ that minorize $A$ and we
have $w_A > w_{A_L}$ and $A_L$ is a minorant in $H$. Since
$A_L\subseteq H$ at iteration $i$, we have
\[w_{F\cup H} < w_{A_L} < w_A < w_{D_U} <
w_{F\cup H}\]  which is a contradiction, and hence the assumption
$\mathcal{X}\neq\{\}$ is false.  The first inequality is because the
algorithm left $A_L$ in $H$ when $F$ was split from $H$, and the
remaining inequalities are due to the above discussion.  Hence the
split at iterations $k+1$ could not have caused a break in
isotonicity.

A similar argument can be made to show that the increased fit for
nodes in $B$ does not cause any isotonic violation.  The proof is
hence completed by induction.
\end{proof}

\begin{figure}[h!] \begin{center}
  \begin{tabular} {c}
     \includegraphics[width=0.49 \textwidth]{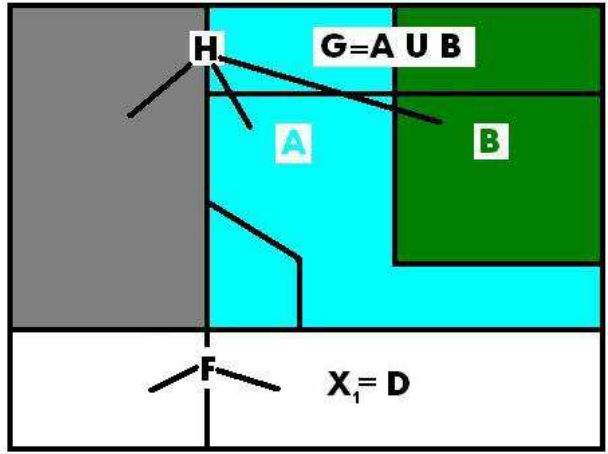}
  \end{tabular}
\caption{Illustration of proof of Theorem
\ref{th:isotonic_solutions} showing the defined sets at iteration
$k+1$. $G$ is the set divided at iteration $k+1$ into $A$ (all blue
area) and $B$ (all green area). The group bordering $A$ from below
denoted by $X_1$ (also referred to as $D$ in the proof) is in
violation with $A$. At iteration $k_0$, $G$ is part of the larger
group $H$ and $X_1$ is part of the larger group $F$.  At iteration
$k_0$, groups $F$ and $H$ are separated. The proof shows that when
$A$ and $B$ are split at iteration $k+1$, no group such as $X_1$
where $w_{X_1}>w_A$ could have existed.  In the picture, $X_1$ must
have been separated at an iteration $k_0<k+1$, but the proof,
through contradiction, shows that this cannot occur. }
\label{fig:proof2}
\end{center} \end{figure}

 \end{document}